\documentclass[a4paper,english]{scrartcl}

\usepackage[utf8]{inputenc}
\usepackage{enumerate,hyperref}
\usepackage[inline]{enumitem}
\usepackage{amsmath,amsthm,amsfonts,amssymb,stmaryrd,mathtools}
\usepackage{etex,etoolbox,thmtools,environ}
\usepackage{verbatim}
\usepackage{xspace}
\usepackage{tikz}
\usetikzlibrary{calc,matrix,arrows,positioning,intersections}

\newif\ifappendix

\appendixtrue
\newcommand{\suppressed}[1]{
	\PackageWarning{}{#1 suppressed on page \thepage\space on input line \number\inputlineno!}}

\ifappendix
	\newcommand{\assumeapx}[1]{#1}
\else
	\newcommand{\assumeapx}[1]{\suppressed{Content assuming the Appendix}}
\fi

\theoremstyle{plain}
\newtheorem{theorem}{Theorem}[section]
\newtheorem{corollary}[theorem]{Corollary}
\newtheorem{proposition}[theorem]{Proposition}
\newtheorem{lemma}[theorem]{Lemma}
\newtheorem{definition}[theorem]{Definition}
\theoremstyle{remark}
\newtheorem{remark}[theorem]{Remark}
\newtheorem{example}[theorem]{Example}

\makeatletter
\ifcsname spn@wtheorem\endcsname
	\providecommand{\@lastarg}[1]{#1}
\else
	\providecommand{\@lastarg}[4]{#4}
\fi
\newcounter{freshthmlabel}

\newif\ifnoproofsketch
\noproofsketchfalse

\newcommand{\setthmlabel}{%
	\stepcounter{freshthmlabel}%
	\label{proofatend:\thefreshthmlabel}%
	\global\noproofsketchtrue
}

\newcommand\fixstatement[2][\proofname\space of]{%
  \ifcsname thmt@original@#2\endcsname
    \AtEndEnvironment{#2}{%
      \xdef\pat@label{\expandafter\expandafter\expandafter
        \@lastarg\csname thmt@original@#2\endcsname}
		\xdef\pat@proofof{\@nameuse{pat@proofof@#2}}%
		\setthmlabel
    }%
  \else
	\ifcsname #2name\endcsname
		\AtEndEnvironment{#2}{%
			\xdef\pat@label{\expandafter\expandafter\expandafter\@lastarg\csname #2name\endcsname}
			\xdef\pat@proofof{\@nameuse{pat@proofof@#2}}%
			\setthmlabel
		}%
	\else
		\AtEndEnvironment{#2}{%
			\xdef\pat@label{\expandafter\expandafter\expandafter\@lastarg\csname #1\endcsname}
			\xdef\pat@proofof{\@nameuse{pat@proofof@#2}}%
			\setthmlabel
		}%
    \fi
  \fi
  \@namedef{pat@proofof@#2}{#1}%
}

\globtoksblk\prooftoks{100}
\newcounter{proofcount}
\newcounter{printedproofcount}

\newcommand{\seeproofref}{%
	\ifappendix
	\renewcommand*{\qedsymbol}{%
		\hyperref[proofatend:proof-\thefreshthmlabel]{%
			\(\Box\)}}
	\fi
}

\NewEnviron{proofsketch}[1][Proof (Sketch)]{%
	\begin{proof}[#1]		
		\BODY
		\ifappendix
				\seeproofref
		\fi
	\end{proof}
	\global\noproofsketchfalse
}

\NewEnviron{proofatend}{% 
	\edef\next{%
		\noexpand\begin{proof}[{\pat@proofof\space\pat@label\space\noexpand\ref{proofatend:\thefreshthmlabel}}]%
		\noexpand\phantomsection%
		\noexpand\label{proofatend:proof-\thefreshthmlabel}%
		\unexpanded\expandafter{\BODY}}%
		\global\toks\numexpr\prooftoks+\value{proofcount}\relax=\expandafter{\next\end{proof}}
	\stepcounter{proofcount}}
	
	\newcommand{\omittedproofs}[1][{}]{%
		\ifnum\value{printedproofcount}<\value{proofcount}
			#1
			\count@=\value{printedproofcount}
			\loop
				\the\toks\numexpr\prooftoks+\count@\relax
				\stepcounter{printedproofcount}
				\ifnum\count@<\value{proofcount}%
					\advance\count@\@ne
			\repeat
		\fi
	}
\makeatother

\fixstatement{theorem}
\fixstatement{proposition}
\fixstatement{corollary}
\fixstatement{lemma}
\makeatletter
\newcommand{\superimpose}[2]{%
  {\ooalign{$#1\@firstoftwo#2$\cr\hfil$#1\@secondoftwo#2$\hfil\cr}}}
\makeatother

\makeatletter
	\newcommand{\footnoteref}[1]{%
		\protected@xdef\@thefnmark{\ref{#1}}\@footnotemark%
	}
\makeatother

\makeatletter
\newcommand{\customlabel}[2]{%
	\protected@write \@auxout {}{\string \newlabel {#1}{{#2}{\thepage}{#2}{#1}{}} }%
	\hypertarget{#1}{#2}
}
\makeatother

\newcommand{\oversetlabel}[3][\theequation]{
	\stepcounter{equation}
	\phantomsection
	\overset{({\customlabel{#2}{#1}})}{#3}
}
\newcommand{\eqlabel}[2][\theequation]{
	\ensuremath{\oversetlabel[#1]{#2}{=}}
}

\newcommand{\diff}{\mathop{}\!\mathrm{d}}

\newcommand{\defeq}{\triangleq}
\newcommand{\defiff}{\stackrel{\triangle}{\iff}}

\DeclareMathOperator{\obj}{obj}

\newcommand{\cat}[1]{\textnormal{\textsf{#1}}\xspace}
\newcommand{\Cat}{\cat{Cat}}
\newcommand{\Set}{\cat{Set}}
\newcommand{\kl}{\mathcal{K}l}
\newcommand{\Id}{\mathcal{I}d}
\newcommand{\CM}{\mathcal{C\mspace{-2mu}M}}

\newcommand{\op}{^{op}}
\newcommand{\To}{\Rightarrow}
\newcommand{\pos}{\cat{Pos}}
\newcommand{\posj}{\cat{Pos$^{\!\vee\!}$}}
\newcommand{\cpo}{\cat{$\omega$-Cpo}}
\newcommand{\cpoj}{\cat{$\omega$-Cpo$^{\!\vee\!}$}}

\newcommand{\str}{\textsf{str}}
\newcommand{\cstr}{\textsf{cstr}}
\newcommand{\dstr}{\textsf{dstr}}
\newcommand{\kstr}[1]{{#1\textsf{-str}}}

\hypersetup{
	colorlinks=true,
	linkcolor=blue!50!black,
	citecolor=blue!50!black,
	filecolor=blue!50!black,
	urlcolor=blue!50!black,
	pdfpagelayout=SinglePage,
	pdfpagemode=UseOutlines,
	pdftitle={Behavioural equivalences for coalgebras with unobservable moves},
	pdfauthor={Tomasz Brengos and Marino Miculan and Marco  Peressotti},
	pdfsubject={Concurrency Theory},
	pdfkeywords={coalgebraic semantics,	weak behavioural equivalences, saturation semantics}
}

\begin{document}
\title{Behavioural equivalences\\ for coalgebras with unobservable moves}

\author{
	Tomasz Brengos\thanks{Supported by the grant of Warsaw
	University of Technology no.~504P/1120/0136/000.}\\[-.5ex]
	\small\href{mailto:t.brengos@mini.pw.edu.pl}{\tt t.brengos@mini.pw.edu.pl}\\
	\small	Faculty of Mathematics and Information Sciences \\[-.8ex]
	\small	Warsaw University of Technology, Poland
	\\[2ex]
	\begin{tabular}{ccc}
	Marino Miculan\thanks{Partially supported by MIUR project 2010LHT4KM (\emph{CINA}).}&\qquad& Marco  Peressotti\\[-.5ex]
	\small\href{mailto:marino.miculan@uniud.it}{\tt marino.miculan@uniud.it}
	&\qquad&
	\small\href{mailto:marco.peressotti@uniud.it}{\tt marco.peressotti@uniud.it}
	\end{tabular}\\
	\small	Laboratory of Models and Applications of Distributed Systems \\[-.8ex]
	\small	University of Udine, Italy
	\small	Department of Mathematics and Computer Science\\[-.8ex]
}
\date{}

\maketitle

\vspace{-6ex}
\begin{center}
\begin{minipage}{.95\textwidth}
\paragraph{Abstract}
  We introduce a general categorical framework for the definition of
  weak behavioural equivalences, building on and extending recent
  results in the field.  This framework is based on special \emph{order
    enriched categories}, i.e.~categories whose hom-sets are endowed
  with suitable complete orders. Using this structure we provide an
  abstract notion of \emph{saturation}, which allows us to define
  various (weak) behavioural equivalences.  
  We show that the Kleisli categories of many common monads are
  categories of this kind.  On one hand, this allows us to instantiate
  the abstract definitions to a wide range of existing systems
  (weighted LTS, Segala systems, calculi with names, etc.), recovering
  the corresponding notions of weak behavioural equivalences; on the
  other, we can readily provide new weak behavioural equivalences for
  more complex behaviours, like those definable on presheaves,
  topological spaces, measurable spaces, etc.
\end{minipage}
\end{center}

%\tableofcontents

\section{Introduction}

\looseness=-1
Since Aczel's seminal work \cite{am89:final}, the theory of coalgebras
has been recognized as a good framework for the study of concurrent
and reactive systems \cite{rutten:universal}: systems are represented
as maps of the form $X\to BX$ for some suitable \emph{behavioural
  functor} $B$.  By changing the underlying category and the functor
we can cover a wide range of cases, from traditional labelled
transition systems to systems with I/O, quantitative aspects,
probabilistic distribution, stochastic rates, and even systems with
continuous state.  Frameworks of this kind are very
useful both from a theoretical and a practical point of view, since
they prepare the ground for general results and tools which can be
readily instantiated to various cases, and moreover they help us to
discover connections and similarities between apparently different
notions.  In particular, Milner's strong bisimilarity can be
characterized by the final coalgebraic semantics and
\emph{coalgebraic bisimulation}; this has paved the way for the
definition of strong bisimilarity for systems with peculiar
computational aspects and many other important results (such as Turi
and Plotkin's bialgebraic approach to \emph{abstract GSOS}
\cite{tp97:tmos}).
More recently, Hasuo et al.~\cite{hasuo07:trace} have showed that, when
the functor $B$ is of the form $TF$ where $T$ is a monad, the
\emph{trace equivalence} for systems of the form $X\to TFX$ can be
obtained by lifting $F$ to the Kleisli category of $T$.  This has led
to many results connecting formal languages, automata theory and
coalgebraic semantics
\cite{silvawesterbaan2013:calco,sbbr:lmcs13,bonchi2015killing,BBHPRS14}.

These remarkable achievements have boosted many attempts to cover other
equivalences from van Glabbeek’s spectrum \cite{glabbeek90:spectrum}.
However, when we come to behavioural equivalences for systems with
unobservable (i.e., internal) moves, the situation is not as clear.
The point is that what is ``unobservable'' depends on the system: in
LTSs these are internal steps (the so-called ``$\tau$-transitions''), but in
systems with quantitative aspects or dealing with resources internal
steps may still have observable effects.  This has led to many
definitions, often quite \emph{ad hoc}.  Some follow Milner's
``double arrow'' construction (i.e., strong bisimulations of the
system saturated under $\tau$-transitions), but in general this
construction does not work; in particular for quantitative systems we
cannot apply directly this schema,
and many other solutions have been proposed; see e.g.
\cite{mb2007:ictcs,sokolova09:sacs,mb2012:qapl,mp2013:weak-arxiv,gp:icalp2014,doberkat:tamc2008}.
In non-de\-ter\-min\-is\-tic probabilistic systems, for example, the
counterpart of Milner's weak bisimulation is Segala's weak bisimulation
\cite{sl:njc95}, which differs from Baier-Hermann's
\cite{baier97:cav}.

This situation points out the need for a general, uniform framework
covering many weak behavioural equivalences at once.  This is the
problem we aim to address in this paper.  Analysing previous work in
this direction \cite{mp2013:weak-arxiv,gp:icalp2014,brengos2013:corr},
a common trait we notice is that systems are \emph{saturated} by
adding unobservable moves to ``fill the gaps'' which can be observed
by a strong bisimulation.  Although different notions of saturations
are used for different weak bisimilarities, they are accumulated by
\emph{circular} definitions; in turn, these definitions can be
described as \emph{equations} in a suitable domain of
\emph{approximants} and solved taking advantage of some fixed point
theory.  Different equations and different domains yield different
saturations and hence different notions of weak bisimilarity.

In the wake of these observations, in this paper we propose to host
these constructions in \emph{order-enriched categories} whose hom-sets
are additionally endowed with \emph{binary joins} for ``merging'' approximants, and
a \emph{complete order} to guarantee convergence of approximant
chains.  In this setting we can define, and solve, the abstract
equations corresponding to many kinds of weak observational
equivalence.  For example, we will show the abstract schemata
corresponding to Milner's and Baier-Hermann's versions of weak
bisimulations; hence, these two different bisimulations are
applications of the same general framework.  Then, we show that the
Kleisli categories of many monads commonly used for defining
behavioural functors meet these mild requirements; this allows us to
port the definitions above to a wide range of behaviours in different
categories (such as presheaf categories, topological spaces and
measurable spaces).

Another aspect highlighted by these applications is that the notion of
``unobservability'' can be considered as a specific computational
effect embedded in the monad structure presented by the behavioural
functors.  This observation fosters the idea that unobservability is
orthogonal to saturation, in the sense that we can develop a theory of
saturation modularly in the notion of unobservability: which aspects
are ignored in the behavioural equivalence can be decided by choosing
the suitable unobservation monad.

\paragraph{Synopsis}
Section~\ref{sec:prelim} provides some preliminaries about
coalgebras, monads and order enriched categories.  
Section~\ref{sec:2_cat_perspective} contains the main
contribution of this work, i.e.~a 2-categorical perspective
on weak behavioural equivalence for coalgebras whose behaviours
have an unobservable part.
In Section~\ref{sec:examples} we provide a few representative examples
of behaviours covered by our results.
Some conclusions and directions for further work are in
Section~\ref{sec:concl}.
\assumeapx{Longer proofs are in Appendix~\ref{proofatend:proofs}.}

We assume the reader to be familiar with the theory of 
coalgebras and behavioural equivalences; for an introduction, we refer to \cite{adamek05:intro_coalg}.

\section{Preliminaries}\label{sec:prelim}

\subsection{Monads}
A monad on a category \cat{C}is a triple $(T,\mu,\eta)$ where
$T$ is an endofunctor over \cat{C}and 
$\mu \colon TT \To T$ and $\eta \colon \Id \To T$
are two natural transformations such that the diagrams below
commute:
\[
	\begin{tikzpicture}[
		auto, xscale=1.3, font=\small,
		nat/.style={double,-implies},
		baseline=(current bounding box.center)]	
		\node (n0) at (0,1) {\(T^3\)};
		\node (n1) at (1,1) {\(T^2\)};
		\node (n2) at (0,0) {\(T^2\)};
		\node (n3) at (1,0) {\(T\)};
		\draw[nat] (n0) to node {\(\mu T\)} (n1);
		\draw[nat] (n0) to node[swap] {\(T\mu\)} (n2);
		\draw[nat] (n1) to node {\(\mu\)} (n3);
		\draw[nat] (n2) to node[swap] {\(\mu\)} (n3);
	\end{tikzpicture}
	\qquad
	\begin{tikzpicture}[
		auto, xscale=1.3, font=\small,
		nat/.style={double,-implies},
		baseline=(current bounding box.center)]	
		\node (n0) at (0,1) {\(T\)};
		\node (n1) at (1,1) {\(T^2\)};
		\node (n2) at (2,1) {\(T\)};
		\node (n3) at (1,0) {\(T\)};
		\draw[nat] (n0) to node {\(\eta T\)} (n1);
		\draw[nat] (n0) to node[swap] {\(\Id\)} (n3);
		\draw[nat] (n2) to node[swap] {\(T\eta\)} (n1);
		\draw[nat] (n1) to node {\(\mu\)} (n3);
		\draw[nat] (n2) to node[] {\(\Id\)} (n3);
	\end{tikzpicture}	
\]
These are the coherence conditions of an associative
operation with a unit. In fact, monads over $\cat{C}$ are monoids
in the category of endofunctors over $\cat{C}$ w.r.t.~the monoidal
structure defined by endofunctor composition. Hence,
$\mu$ and $\eta$ are called \emph{multiplication} 
and \emph{unit} of $T$, respectively.

Each monad $(T,\mu,\eta)$ gives rise to a canonical category called \emph{Kleisli
category of $T$} and denoted by $\kl(T)$. This category has the same
objects of the category $\cat{C}$ underlying $T$; its hom-sets are given as 
$\kl(T)(X,Y) = \cat{C}(X,TY)$
for any two objects $X$ and $Y$ and its composition as $g \bullet f = \mu_Z \circ Tg \circ f$ for any two morphisms $f$ and $g$ with suitable domain and codomain
($\bullet$ and $\circ$ denote composition in $\kl(T)$ and $\cat{C}$
respectively).

\begin{example}
The powerset functor $\mathcal P$ admits a monad structure 
$(\mathcal{P},\mu,\eta)$ where
$\eta_X(x) = \{x\}$ and $\mu_X(Y) = \bigcup Y$
for any set $X$.
\end{example}
\begin{example}
The probability distribution functor $\mathcal{D}$ assigns to any set $X$ the set 
$\mathcal{D}X \defeq \{{\phi\colon X\to [0,1]}\mid \sum_{x\in X} \phi(x) = 1\}$ 
of discrete measures and to any function $f\colon X\to Y$ the function
$\mathcal{D}f(\phi)(y) \defeq \sum_{f(x) = y}\phi(x)$;
it admits a monad structure whose unit and multiplication
are given on each $X$ as $\eta_{X}(x) = \delta_x$ and 
$\mu_X(\phi)(x) = \sum_{\psi}\psi(x)\cdot\phi(\psi)$.
\end{example}

\paragraph{Commutative monads}
Assume $(\cat{C},\otimes,I)$ to be a monoidal category and let
\begin{gather*}
	\alpha = \{\alpha_{X,Y,Z}\colon (X \otimes Y) \otimes Z \to X \otimes (Y \otimes Z) \}_{X,Y,Z \in \cat{C}}\\
	\lambda = \{\lambda_{X}\colon I \otimes X \to X \}_{X \in \cat{C}}\qquad
	\rho = \{\rho_{X}\colon X \otimes I \to X \}_{X \in \cat{C}}
\end{gather*}
denote its associator, left unitor, and right unitor.
A monad $(T,\mu,\eta)$ on \cat{C} is called
\emph{strong} if it is equipped with a family of morphisms
$$\{\str_{X,Y} \colon X \otimes TY \to T(X \otimes Y)\}_{X,Y \in \cat{C}},$$
called \emph{(tensorial) strength}, which is natural in both components 
and is coherent with the structure of monads and monoidal categories, i.e.:
\begin{gather*}
	\mu_{X \otimes Y} \circ T(\str_{X,Y})\circ\str_{X,TY} = 
	\str_{X,Y} \circ (id_X \otimes \mu_Y)
	\qquad
	\eta_{X\otimes Y} = \str_{X,Y} \circ (id_{X} \otimes \eta_Y)
	\\
	T(\alpha_{X,Y,Z}) \circ \str_{X \otimes Y, Z} = 
	\str_{X,Y\otimes Z} \circ id_X \otimes \str_{Y,Z} \circ\alpha_{X,Y,TZ}
	\qquad
	\lambda_{TX} = T(\lambda_X) \circ \str_{I,X} 
	\text{.}
\end{gather*}
Dually, a \emph{costrength} for a monad is family
$\{\cstr_{X,Y} \colon TX \otimes Y \to T(X\otimes Y)\}_{X,Y \in \cat{C}}$
which is natural in both $X$ and $Y$ and 
coherent w.r.t.~the structure of $T$ and $\cat{C}$.
Every strong monad on a symmetric monoidal category
has a costrength given on each component as
$
	\cstr_{X,Y}	\defeq T\gamma_{Y,X} \circ \str_{Y,X} \circ \gamma_{TX,Y}
$
where $\gamma = \{\gamma_{X,Y}\colon X \otimes Y \to Y \otimes X \}_{X,Y \in \cat{C}}$ is the braiding natural isomorphism for $(\cat{C},\otimes,I)$.
A strong monad on a symmetric monoidal category is called 
\emph{commutative} whenever:
\[
	\mu_{X \otimes Y} \circ T(\str_{X,Y}) \circ \cstr_{X,TY} = 
	\mu_{X \otimes Y} \circ T(\cstr_{X,Y}) \circ \str_{TX,Y}
	\text{.}
\]
Every strong commutative monad is a symmetric monoidal monad
(and \emph{vice versa}). In fact, its \emph{double strength} 
$\{\dstr_{X,Y} \colon TX \otimes TY \to T(X \otimes Y)\}_{X,Y \in \cat{C}}$
can be defined in terms of its (co)strength as follows:
\[\dstr_{X,Y} = \mu_{X \otimes Y} \circ T(\str_{X,Y}) \circ \cstr_{X,TY} = 
	\mu_{X \otimes Y} \circ T(\cstr_{X,Y}) \circ \str_{TX,Y}\text{.}\]
Conversely, $\str_{X,Y} = \dstr_{X,Y} \circ (\eta_X \otimes id_{TY})$ and
$\cstr_{X,Y} = \dstr_{X,Y} \circ (id_{TX} \otimes \eta_{Y})$.
We refer the reader to \cite{kock1972strong} for further details on strong
commutative monads.

\begin{example}
	The powerset monad is strong and commutative \cite{hasuo07:trace}.
	Its strength, co\-strength and double strength are given, on each component, as
	follows:
	\[
	\str_{X,Y}(x,Y') = \{x\} \times Y' \qquad
	\cstr_{X,Y}(X',y) = X' \times \{y\} \qquad
	\dstr_{X,Y}(X',Y') = X' \times Y' \text{.}
	\]
\end{example}

\begin{example}
	The probability distribution monad is strong and commutative \cite{hasuo07:trace}:
	\begin{gather*}
		\str_{X,Y}(x,\psi)(x',y') = \delta_x(x')\cdot\psi(y') \qquad
		\str_{X,Y}(\phi,y)(x',y') = \phi(x')\cdot\delta_y(y')\\
		\dstr_{X,Y}(\phi,\psi)(x',y') = \phi(x')\cdot\psi(y') \text{.}
	\end{gather*}
\end{example}

\subsection{Coalgebras and bisimulation} \label{subsec:bisimulations}

Let $\cat{C}$ be a category and let $F \colon \cat{C}\to \cat{C}$ be a functor. 
\begin{definition}
	An \emph{$F$-coalgebra} is a morphism $\alpha\colon X\to FX$ in 
	$\cat{C}$. The domain $X$ of $\alpha$ is	called \emph{carrier}
	and the morphism $\alpha$ is sometimes also called \emph{structure}.
	A \emph{homomorphism} from an $F$-coalgebra $\alpha\colon X\to FX$ to an
	$F$-coalgebra $\beta\colon Y\to~FY$ is a morphism $f \colon X\to Y$ in
	$\cat{C}$ such that $ F(f)\circ \alpha =  \beta \circ f$. The 
	category of $F$-coalgebras and their homomorphisms is denoted by
	$\cat{C}_F$. 
\end{definition}
Many important types of transition systems can be captured using coalgebras.  We will now list two examples here and  discuss others in the remaining part of the paper.

\begin{example}[Labelled transition systems]
	LTSs with labels in a given set $A$ (see e.g~\cite{sangiorgi2011:bis}) can be viewed as coalgebras for the endofunctor $\mathcal{P}(A\times \Id)\colon \Set\to \Set$ \cite{rutten:universal}, where $\mathcal{P}\colon \Set\to \Set$ denotes the powerset functor.
\end{example}
\begin{example}[Fully-probabilistic systems]
	Fully-probabilistic systems \cite{baier97:cav} are modelled as 
	coalgebras for the endofunctor $\mathcal{D}(A \times \Id)$ on \Set 
	\cite{sokolova11} where, $\mathcal{D}$ denotes the probability distribution functor.
\end{example}

Notions of strong bisimulation have been well captured coalgebraically
\cite{am89:final,rutten:universal,staton11}. In this paper we consider
a variant called \emph{kernel bisimulation} i.e.~``a relation which is
the kernel of a common compatible refinement of the two systems''
\cite{staton11}. To be more precise, a relation $R$ (i.e.~a jointly monic span in $\cat C$) is a
\emph{kernel bisimulation}
between two $F$-coalgebras $\alpha_1\colon X_1\to FX_1$ and $\alpha_2\colon X_2\to~FX_2$ if
there are $\beta\colon Y\to FY$ and a cospan of homomorphisms
${\alpha_1 \rightarrow \beta \leftarrow \alpha_2}$ for which $R$ is
the pullback of ${X_1 \rightarrow Y \leftarrow X_2}$ in $\cat{C}$.

In this paper we will consider $T$-coalgebras where $T$ is part of
a monad $(T,\mu,\eta)$ on \cat{C}; in this setting, coalgebras can be seen
as endomorphisms of $\kl(T)$. Therefore, we shall reformulate the notion of kernel bisimulation in 
terms of strong (resp.~weak) \emph{behavioural morphisms} of endomorphisms. 
This choice allows us to streamline
the exposition while working at the abstract level of morphism
between endomorphisms and accommodate saturation even when carried
out in other categories. Moreover, for simplicity's sake, we will consider
strong (resp.~weak) behavioural morphisms on single systems only.
The restriction to single systems aligns us to other recent works on
weak bisimulations \cite{brengos2013:corr,gp:icalp2014}.
We do this without loss of generality as all results presented in this paper readily extend to 
the case of multiple systems. In fact, all categories of coalgebras considered
in this paper have binary coproducts allowing us to treat pairs of systems
as a single one that ``runs them in parallel'' and thus mimic 
the scenario of Staton's kernel bisimulation \cite{staton11}.

\begin{definition}
\label{def:strong-behavioural-equivalence}
	Let $J$ be a subcategory of \cat{K} with all objects of $\cat K$.
	A morphism $i \colon X \to Y \in J$ is a 
	\emph{strong behavioural morphism}
	for an endomorphism $\alpha\colon X \to~X \in \cat{K}$ whenever there 
	is an endomorphism $\beta\colon Y\to Y\in \cat K$
	such that $i \circ \alpha = \beta \circ i$ 
	i.e.~it extends to a homomorphism of endomorphisms from $\alpha$ to $\beta$.
	A relation $R \rightrightarrows X$ in $J$ on the carrier of $\alpha$
	is a \emph{strong behavioural equivalence}
	(or simply \emph{bisimulation}) if it is the kernel pair of some
	strong behavioural morphism for $\alpha$.
\end{definition}
In particular, when $J = \cat{C}$ and $\cat{K} = \kl(T)$
strong behavioural equivalence coincides with kernel bisimulation on a $T$-coalgebra $\alpha:X\to TX$ viewed as an endomorphism in the Kleisli category of $T$.

\subsection{Order enriched categories} \label{subsec:order_enriched}

\begin{definition}
	Let $(\cat{V},\otimes, I)$ be a monoidal category.
	A \emph{category enriched over
	$\cat V$} (or $\cat V$-enriched)
	\cat{C}
	consists of the following data:
	\begin{enumerate*}[label=(\alph*)]
		\item
		a class $\obj(\cat{C})$ of objects of
		\cat{C},
		\item
		an object $\cat{C}(X,Y) \in \cat{V}$
		for each $X,Y \in \obj(\cat{C})$,
		\item
		a map $id_X \colon I \to \cat{C}(X,X)$ in $\cat{V}$
		determining an identity for every $X \in \cat{C}$,
		\item
		and a map $\circ_{X,Y,Z} \colon \cat{C}(Y,Z)\otimes\cat{C}(X,Y) \to \cat{C}(X,Z)$
		in \cat{V} determining a composition for each 
		$X,Y,Z \in \cat{C}$ and such that it is associative
		and has the assignments above as identities, i.e.:
	\end{enumerate*}
		\begin{gather*}
			\lambda_{\cat{C}(X,Y)} = \circ_{X,Y,Y} \bullet \left(id_{Y} \otimes id_{\cat{C}(X,Y)}\right)
			\qquad
			\rho_{\cat{C}(X,Y)} = \circ_{X,X,Y} \bullet \left(id_{\cat{C}(X,Y)}\otimes id_{X}\right)
			\\
			\circ_{W,X,Z} \bullet \left({\circ_{X,Y,Z}} \otimes id_{\cat{C}(W,X)}\right)
			= 
			\circ_{W,Y,Z} \bullet \left(id_{\cat{C}(Y,Z)} \otimes \circ_{W,X,Y}\right) \bullet \alpha_{\cat{C}(Y,Z),\cat{C}(X,Y),\cat{C}(W,X)}
		\end{gather*}
		where $\bullet$ denotes composition in \cat{V} and
		$\lambda$, $\rho$, and $\alpha$
		are the left unitor, the right unitor, and the associator
		of $(\cat{V},\otimes,I)$, respectively.
\end{definition}
Let $(\pos,\times,1)$ be the cartesian category 
of partial orders as objects and monotonic maps as morphisms. We will often refer to \pos-enriched categories as \emph{order enriched}. The enrichment in an order enriched category $\cat{C}$  guarantees that the composition is monotonic in each component i.e.:
\[
f \leq f' \implies g\circ f \leq g\circ f' \land f \circ h \leq f' \circ h.
\]
Many examples considered in this paper are categories with a stronger type of enrichment: \cpo-, \posj- or \cpoj-enrichment. Here, $(\cpo,\times,1)$ is the cartesian category with partial orders that admit  suprema of ascending $\omega$-chains as objects and maps preserving such suprema as morphisms. \posj and \cpoj are full subcategories of $\cat{Pos}$ and \cpo, respectively, in which all objects additionally admit binary joins. It should be noted here that each of the four categories \pos, \posj, \cpo, \cpoj is enriched over itself (with the order on hom-sets imposed by the pointwise order).

The morphisms in \posj and \cpoj do not necessarily preserve binary joins. In other words, it is not necessarily the case that a \posj- or \cpoj-enriched category satisfies right or left distributivity laws. 
\begin{definition}
	We say that  an order enriched category whose hom-sets
	admit binary joins is \emph{right (resp.~left) distributive}
	given that the following hold for arbitrary morphisms 
	$f,g,h,k$ with suitable domain and codomain, respectively: 
	\begin{align}
	\tag{RD} \label{law:RD} (f\vee g) \circ h = f\circ h \vee g\circ h,\\
	\tag{LD} \label{law:LD} k\circ (f\vee g) = k\circ f \vee  k\circ g.
	\end{align}
	We may also say that such equations are satisfied by a category whenever
	they are satisfied by any morphism.
\end{definition}

In general, in \posj- and \cpoj-enriched categories, for any morphisms $f,g,h,k$ with suitable domain and codomain as above  we have:
\[
	(f\vee g) \circ h \geq  f\circ h \vee g\circ h 
	\text{ and } k\circ (f\vee g) \geq k\circ f \vee  k\circ g.
\] 
For any order enriched category \cat{K} let $\cat{K}\op$ denote the dual category with the order on hom-sets left intact. Note that if \cat{K} is left (resp.~right) distributive then $\cat{K}\op$ is right (resp.~left) distributive.

In this paper we shall use several 2-categorical notions in the setting
of order enriched categories (see e.g.~\cite{lack:09,lack:morita:2014,street:1972} for a general 2-categorical perspective). We recall some of these notions here.

A \emph{monad} in an order enriched category is an endomorphism $\sigma\colon X\to X$ satisfying $id_X\leq \sigma$ and $\sigma\circ \sigma\leq \sigma$. A \emph{lax functor} $F$ from a category  $\cat{C}$ to an order enriched category  \cat{K} consists of the following assignments. For any object $X\in \cat{C}$ there is an object $FX \in \cat{K}$ and for any morphism $f\colon X\to Y$ a morphism $F(f)\colon FX\to FY$ such that for any composable morphisms $f,g\in \cat{C}$ we have:
$$
F(f)\circ F(g)\leq F(f\circ g) \text{ and } id_{FX}\leq F(id_X).
$$
Given two lax functors $F,G\colon \cat{C}\to \cat{K}$ a family of morphisms $\psi = \{\psi_X\colon FX\to GX\}_{X\in \cat{K}}$ 
in $\cat{K}'$ is \emph{lax natural transformation} between $F$ and $G$ provided that for any $f\colon X\to Y$ in $\cat{C}$ we have:
\[
	\begin{tikzpicture}[
		auto, scale=1, font=\small,
		baseline=(current bounding box.center)]	
		\matrix (m) [matrix of math nodes, row sep=4ex, column sep=5ex]{ 
			FX & GX \\
			FY & GY \\
		};
		\draw[->] (m-1-1) to node {\(\psi_X\)} (m-1-2);
		\draw[->] (m-1-1) to node[swap] {\(Ff\)} (m-2-1);
		\draw[->] (m-1-2) to node {\(Gf\)} (m-2-2);
		\draw[->] (m-2-1) to node[swap] {\(\psi_Y\)} (m-2-2);
		\node at ($(m-1-1)!.5!(m-2-2)$) {\(\geq\)};
	\end{tikzpicture}
\]
We define \emph{oplax functors} and \emph{oplax transformations} by reversing the order in the above.
Finally, a \emph{$\cat{C}$-adjunction} consists of a functor $F\colon \cat{C}\to \pos$ equipped with a family of adjunctions $\phi_X \dashv \nu_X \colon FX \to P_X$ between posets $FX$ and $P_X$ considered as categories for any  object $X\in \cat{C}$.  
Akin to order theory, monads and adjunctions in order-enriched categories are also called \emph{closure operators} and \emph{Galois connections}.

\subsection{Coalgebras with unobservable moves}
 \label{sec:monadic-observables} 
Originally
\cite{hasuojacobssokolova2006:jssst,silvawesterbaan2013:calco},
coalgebras with silent actions were introduced in the context of
coalgebraic trace semantics as $T(F+\Id)$-coalgebras for a monad $T$ and an endofunctor $F$ on $\cat{C}$
(e.g.~$\mathcal{P}((A\times\Id + \Id) \cong
\mathcal{P}((A+\{\tau\})\times\Id)$).  Intuitively,
$T$, $F$, and $\Id$ describe the ``branching'',
``observable'', and ``unobservable'' computational aspects
respectively with the monad defining how moves are concatenated.  For
many systems modelled by coalgebras in \Set observables are actually
given by a set of labels (or alphabet) $A$. In this case:
\[
	T(F+\Id)=T(A\times \Id + \Id)\cong T(A_\tau\times \Id),
\]
where $A_\tau \defeq A+\{\tau\}$ for short. However, this readily generalises
to more complex notions of observables and beyond the 
category \Set; e.g.,
the alphabet may be given as measurable space (cf.~FlatCCS \cite{bm:2015stocsos}) 
or as a presheaf (cf.~$\pi$-calculus \cite{sangiorgi2003pi}).

\looseness=-1
In this subsection we show how unobservable moves can be endowed with a
suitable monadic structure whose Kleisli category is \cpoj-enriched (under mild assumptions).  We build on the approach introduced in
\cite{brengos2013:corr} where Brengos showed that given some mild
assumptions on $T$ and $F$ we may either introduce a monadic structure
on a lifting $\overline{F}+\Id$ of the functor
$F+\Id$ to the Kleisli category $\kl(T)$ or embed the lifting
$\overline{F}+\Id$ into the monad $\overline{F}^{\ast}$, where
$\overline{F}^{\ast}$ denotes the free monad over the lifting
$\overline{F}\colon \kl(T)\to \kl(T)$ of $F$.  The monadic structure of
(un)observables on $\kl(T)$ is composed with $T$ (along the adjoint situation
on \cat{C} and $\kl(T)$) yielding a monadic structure on $T(F+\Id)$
and $TF^{\ast}$, respectively.

This result corroborates the view of \emph{unobservables as a computational
effect}, along the line of Moggi's theory, and plays a fundamental r\^ole
in our treatment of systems with unobservable moves: it allows us
\emph{not} to specify the invisible moves explicitly and let a
monadic structure of the behavioural functor handle them
internally.  Instead of considering $T(F+\Id)$-coalgebras we
consider $T'$-coalgebras for a monad $T'$ on an arbitrary category
$\cat{C}$.

\paragraph{Liftings to and monads on $\kl(T)$}
Let $(T,\mu,\eta)$ be a monad on $\cat{C}$.
There is an inclusion functor $(-)^\sharp\colon \cat{C} \to \kl(T)$
mapping any $X$ to itself and any $f$ to 
$\eta_Y\circ f$.
This functor is the left adjoint to the forgetful functor 
$U_T \colon \kl(T) \to \cat{C}$, mapping any $X$ to
$TX$ and any $f$ to its \emph{Kleisli extension} $\mu_Y \circ Tf$.

We say that a functor $F\colon \cat{C}\to\cat{C}$ \emph{lifts to} a functor $\overline{F}\colon \kl(T)\to\kl(T)$ provided that the following diagram commutes:
\[
	\begin{tikzpicture}[
			auto, scale=1.4,font=\small,
			baseline=(current bounding box.center)]	
		\matrix (m) [matrix of math nodes, row sep=4ex, column sep=5ex]{ 
			\kl(T) & \kl(T) \\
			\cat{C}& \cat{C}\\
		};
		\draw[->] (m-1-1) to node {\(\overline F\)} (m-1-2);
		\draw[->] (m-2-1) to node[swap] {\(F\)} (m-2-2);
		\draw[->] (m-2-1) to node {\((-)^\sharp\)} (m-1-1);
		\draw[->] (m-2-2) to node[swap] {\((-)^\sharp\)} (m-1-2);
	\end{tikzpicture}	
\]
Given a functor $F\colon \cat{C}\to\cat{C}$ there is a one-to-one correspondence between its liftings $\overline{F}\colon \kl(T)\to\kl(T)$ and \emph{distributive laws}  $\lambda\colon FT\To TF$ between the functor $F$ and the monad $T$ (see e.g.~\cite{mulry:mfps1993} for a detailed definition and properties). Given a distributive law $\lambda\colon FT\To TF$ we define $\overline{F}\colon \kl(T)\to\kl(T)$ on any object $X$ and any morphism
$f \colon X \to Y$ in $\kl(T)$ as:
\[
	\overline{F}X \defeq FX 
	\qquad
	\overline{F}f \defeq \lambda_Y \circ Ff\text{.}
\]
Conversely, a lifting $\overline{F}\colon \kl(T)\to \kl(T)$ gives rise to $\lambda\colon FT\To TF$ defined on each component $\lambda_X$ as $\overline{F}(id_{TX})$. 

In \cite{hasuo07:trace} Hasuo et al.~showed that \emph{polynomial
functors} (there called ``shapely'') 
admit canonical distributive laws (hence liftings) if the
monad is commutative (w.r.t.~cartesian products) and $\cat{C}$ has
countable coproducts (hence, so does $\kl (T)$ as coproducts are lifted from $\cat{C}$). A functor is said to be \emph{polynomial} whenever it is described by the
grammar:
\[
	F\Coloneqq \Id \mid A \mid F_1 \times F_2 \mid \textstyle\coprod_{i \in I} F_i
\]
where $A \in \cat{C}$ and $I \subseteq \mathbb N$.  The canonical
distributive law is defined by structural recursion.  Note that a
strong monad (not necessarily commutative)
suffices to deal with the classic case of ``labels with silent actions'',
i.e.~$A \times \Id + \Id$.
\begin{example}\label{example:monads_on_Kl_P}
	Take $\cat{C}=\Set$ and $T = \mathcal P$. 
	Since the powerset monad is strong and commutative, 
	the functor $A \times \Id\colon \Set \to \Set$ lifts to
	$\overline{A\times \Id}\colon \mathcal{K}l(\mathcal{P})\to \mathcal{K}l(\mathcal{P})$ acting as $A \times \Id$ on objects
	and as $\mathsf{dstr}(\eta_A \times \Id)$ on morphisms. In particular,
	for any object $X\in \mathcal{K}l(\mathcal{P})$ and any morphism $f\colon X\to Y$ in $\mathcal{K}l(\mathcal{P})$ we have:
	\[
		(\overline{A\times \Id}) X = A \times X 
		\quad\text{ and }\quad
		(\overline{A\times \Id})f(a,x) = \{(\sigma,y)\mid y\in f(x)\}.
	\]
\end{example}

If the lifting of a functor presents a monad structure in $\kl(T)$, say
$(\overline S,\nu,\theta)$, we have the following two adjoint
situations:
\begin{equation}
\label{diag:kl-t-s-adj}
	\begin{tikzpicture}[font=\small,
		baseline=(current bounding box.center)]
		\node (n0) {\(\cat{C}\)};
		\node[right=6ex of n0] (n1) {\(\kl(T)\)};
		\node[right=6ex of n1] (n2) {\(\kl(\overline S)\)};
		\draw[right hook->,bend left] (n0.north east) to node [above] {\((-)^\sharp\)} (n1.north west);
		\draw[->,bend left] (n1.south west) to node [below] {\(U_T\)} (n0.south east);
		\node[rotate=-90] at ($(n0.east)!.5!(n1.west)$) {\(\dashv\)};
		\draw[right hook->,bend left] (n1.north east) to node [above] {\((-)^\sharp\)} (n2.north west);
		\draw[->,bend left] (n2.south west) to node [below] {\(U_{\overline S}\)} (n1.south east);
		\node[rotate=-90] at ($(n1.east)!.5!(n2.west)$) {\(\dashv\)};
	\end{tikzpicture}	
\end{equation}
The adjoint situation defined by their composition endows $TS\colon \cat{C}\to \cat{C}$
with a monad structure whose multiplication and unit are defined as:
\begin{equation*}
	\mu_{S}\circ T\mu_{S} \circ T^2 \nu \circ T\lambda_{S} \qquad \text{ and } \qquad \theta
	\text{.}
\end{equation*}
For any $f\colon X\to Y$ and $g\colon Y\to Z$ in $\kl(TS) = \kl(\overline{S})$ their
composite $g \bullet f$ is given as:
\begin{equation}
	\label{diag:kl-s-composition}
	\begin{tikzpicture}[
			auto, xscale=1.8,font=\small,
			baseline=(current bounding box.center)]	
		\node (n0) at (.2,0) {\(X\)};	
		\node (n1) at (1,0) {\(TSY\)};	
		\node (n2) at (2,0) {\((TS)^2Z\)};	
		\node (n3) at (3,0) {\(T^2S^2Z\)};	
		\node (n4) at (4,0) {\(T^3SZ\)};	
		\node (n5) at (5,0) {\(T^2SZ\)};	
		\node (n6) at (6,0) {\(TSZ\)};
		
		\draw[->,rounded corners] (n0) -- +(0,.7) -| node[pos=.25] {\(g \bullet f\)} (n6);
		\draw[->] (n0) to node[swap] {\(f\)} (n1);
		\draw[->] (n1) to node[swap] {\(TSg\)} (n2);
		\draw[->] (n2) to node[swap] {\(T\lambda_{SZ}\)} (n3);
		\draw[->] (n3) to node[swap] {\(T^2\nu_Z\)} (n4);
		\draw[->] (n4) to node[swap] {\(T\mu_{SZ}\)} (n5);
		\draw[->] (n5) to node[swap] {\(\mu_{SZ}\)} (n6);
	\end{tikzpicture}	
\end{equation}

The adjoint situations in \eqref{diag:kl-t-s-adj} allow us to compose 
several computational aspects and, when each layer 
is well-behaved w.r.t.~the underlying enrichment, to lift enrichment too. 
Here, being well-behaved means that
the functor of a monad $(\overline{S},\nu,\theta)$ on the \cpo-enriched
$\kl(T)$ is \emph{locally continuous}, i.e.~it preserves suprema
of ascending $\omega$-chains:
\[\textstyle
	\overline S\big(\bigvee_{i < \omega} f_i\big) = \bigvee_{i <\omega}\left(\overline S f_i\right)\text{.}
\] 
Note that $S$ may not be a monad, we only assume that its lifting is.
Remarkably, polynomial\footnote{%
	Analytic functors are a class of endofunctors over \Set 
	containing polynomial ones and have locally continuous liftings
	\cite{milius:calco2009anfun}.
} functors have locally continuous liftings \cite{hasuo07:trace}.
\begin{theorem}
	\label{thm:t_kleisli_s_parameterised_saturation}
Assume that $\kl(T)$ is \cpoj-enriched. If $\overline{S}$ is locally continuous then $\kl(\overline S)=\kl(TS)$ is \cpoj-enriched. If $\kl(T)$ is \eqref{law:LD} then so is $\kl(\overline{S})$.
\end{theorem}
\begin{proof} %atend}
	The partial order on $\kl(\overline S)(X,Y)$ is imposed by the one on $\kl(T)(X,SY)$. Hence, it is clear that the order on hom-sets in $\kl(\overline{S})$  admits binary joins as the order on hom-sets in $\kl(T)$. Now take an ascending chain $f_0\leq f_2 \leq \ldots$ in $\kl(\overline{S})(X,Y))$ and morphisms $g,h$ in $\kl(\overline{S})$ with suitable domain and codomain. Let $\bullet$ and $\circ$ denote the compositions in $\kl(\overline{S})$ and $\kl(T)$, respectively. We have:
	\begin{gather*}
		\textstyle
		g\bullet \bigvee_{n<\omega} f_n 
		= \nu\circ \overline{S} g\circ \bigvee_{n<\omega} f_n 
		= \bigvee_{n<\omega} \nu\circ \overline{S}g \circ f_n 
		= \bigvee_{n<\omega} g\bullet f_n\text{,}
		\\ \textstyle
		\bigvee_{n<\omega} f_n \bullet h 
		= \nu\circ \overline{S} \bigvee_{n<\omega} f_n \circ h 
		= \bigvee_{n<\omega} \nu\circ \overline{S}f_n \circ h 
		= \bigvee_{n<\omega} f_n \bullet h\text{.}
	\end{gather*} 
	This proves that $\kl(\overline{S})$ is \cpoj-enriched. 
	Let $\kl(T)$ be \eqref{law:LD}. Then:
	\begin{gather*}
		g\bullet (f_1\vee f_2) 
		= \nu\circ \overline{S} g\circ (f_1\vee f_2) 
		=  \nu\circ \overline{S}g \circ f_1 \vee \nu\circ \overline{S}g \circ f_2
		= g\bullet f_1\vee g\bullet f_2
	\end{gather*} 
	completing the proof.
\end{proof} %atend}

For the aims of this work $S$ will be a functor modelling internal
moves i.e.~with shape $F + \Id$. The following result states
that these functors extend to monads on $\kl(T)$ whenever the category
has zero morphisms\footnote{%
	We say that a category has \emph{zero morphisms} if for
	any two objects $X,Y$ there is a morphism $0_{X,Y}$ which is
	an annihilator w.r.t.~composition, i.e.: $f\circ 0 = 0=0\circ g$
	for any morphisms $f,g$ with suitable domain and codomain.}.
\begin{theorem}[\hspace{.01em}\cite{brengos2013:corr}]
	\label{thm:monadic-fg-zero-morph}
	Let $\overline F$ be a lifting of $F$ to $\kl(T)$. 
	If $\cat{C}$ has binary coproducts and $\kl(T)$ has zero morphisms then
	$\overline{F+\Id}=\overline{F}+\Id$ extends to a monad on $\kl(T)$ whose unit 
	is $\mathsf{inr} \colon \Id \To \overline{F+\Id}$ 
	and whose multiplication is
	\[
	\left[\mathsf{inl},id_{\overline{F+\Id}}\right]\circ \left(\overline{F}([0+id])+id_{\overline{F+\Id}}\right)\colon  (\overline{F+\Id})^2 \To \overline{F+\Id}\text{.}\]\par
\end{theorem}

\begin{example}\label{example:two_monads_on_Kl_P}
Let $T=\mathcal{P}$ and let $A$ be an arbitrary set. By Example~\ref{example:monads_on_Kl_P} the endofunctor $A_\tau \times \Id$ over \Set lifts to $\mathcal{K}l(\mathcal{P})$.  By Theorem~\ref{thm:monadic-fg-zero-morph} its lifting $\overline{A_\tau\times \Id}\cong \overline{A\times \Id}+\Id$ can be equipped with a monadic structure $(\overline{A_\tau\times \Id},\nu,\theta)$ whose multiplication and unit are given on their components by:
\[
	\theta_X(x) = \{(\tau,x)\} 
	\quad\text{ and }\quad 
	\nu_X(a,b,x) = 
	\begin{cases}
		\{(a,x)\}&\text{ if } b=\tau \\
		\{(b,x)\}&\text{ if } a=\tau \\
		\emptyset&\text{ otherwise.}
	\end{cases}
\]
Now, the functor $\mathcal{P}(A_\tau \times \Id)$
carries a monadic structure which is a consequence of composing two adjunctions $\Set\rightleftarrows \mathcal{K}l(\mathcal{P})\rightleftarrows \mathcal{K}l(\overline{A_\tau\times \Id})$ as described earlier in this subsection. The composition in $\kl(\mathcal{P}(A_\tau\times \Id))$ is given as follows. For $f\colon X\to \mathcal{P}(A_\tau\times Y)$ and $g\colon Y\to \mathcal{P}(A_\tau \times Z)$ we have:
\begin{equation}
	\label{eq:lts-kl-comp}
	g\bullet f(x) = \{(a,z) \mid x\xrightarrow{a}_f y \xrightarrow{\tau}_g z 
	\text{ or } x\xrightarrow{\tau}_f y \xrightarrow{a}_g z 
	\text{ for some }y\in Y, a \in A_\tau \}
\end{equation}
See \cite{brengos2013:corr} for details.
\end{example}

We point out that the monad laws for unit and multiplication in the above example generalise the derivation
rules
\begin{align}
	\frac{x \xrightarrow{\tau} y \qquad y \xrightarrow{\tau} z}{x \xrightarrow{\tau} z}
	\qquad
	\frac{x \xrightarrow{a} y \qquad y \xrightarrow{\tau} z}{x \xrightarrow{a} z}
	\qquad	
	\frac{x \xrightarrow{\tau} y \qquad y \xrightarrow{a} z}{x \xrightarrow{a} z}
	\qquad \label{LTS:closure_rules}
\end{align}
describing transitivity of unobservable moves and 
left and right ``absorption'' of unobservable moves by observable ones.
The use of zero morphisms to ``kill'' consecutive observables corresponds
to the absence of rules for this case.

Remarkably, small tweaks in the monad defined by
Theorem~\ref{thm:monadic-fg-zero-morph} allow us to deal with various
interactions between observable and unobservable moves; e.g., we can
cover the groupoidal nature of \emph{reversible computations} by
considering a multiplication sketched by the following 
rules:
\[
	\frac{x \xrightarrow{a} y \qquad y \xrightarrow{a^{-1}} z}{x \xrightarrow{\tau} z}
	\qquad	
	\frac{x \xrightarrow{a^{-1}} y \qquad y \xrightarrow{a} z}{x \xrightarrow{\tau} z}
\]

Hence, from now on the term ``coalgebra with unobservable moves'' becomes synonymous to ``coalgebra whose endofunctor carries a monadic structure''.

\section{A 2-categorical perspective on weak behavioural equivalence}
\label{sec:2_cat_perspective} 

In this section we give a 2-categorical setting for the notion of
weak behavioural equivalence on a coalgebra with unobservable
moves. Our main point of interest are $T$-coalgebras for a
monad $(T,\mu,\eta)$ on a category $\cat{C}$. Any such coalgebra $\alpha\colon X\to TX$
is an endomorphism in $\mathcal{K}l(T)$.  Hence, in order to be as
general as possible, we will now build a theory of weak behavioural
equivalences on endomorphisms of a category \cat{K} satisfying some
extra assumptions, bearing in mind our prototypical example of (a full
subcategory of) $\mathcal{K}l(T)$.

\subsection{Categories of lax preasheaves and oplax transformations}\label{subsection:the_category_K_l}

\begin{definition}[Lax \cpoj-presheaves]\label{def:laxpresh}
  Let \cat{K} be a small \cpoj-enriched category.  A lax functor of
  type $\cat{K}\to \cpoj$ will be called \emph{lax \cpoj-presheaf};
  often the ``\cpoj-'' will be omitted.  Lax presheaves over
  \cat{K} and oplax natural transformations between them form the
  category $\widetilde{\cat{K}}\defeq \cat{Lax}(\cat{K},\cpoj)_{opl}$.
\end{definition}

\begin{remark}
Our aim is to apply this definition for $\cat{K}=\mathcal{K}l(T)$.
Actually, for all monads presented in this paper $\mathcal{K}l(T)$
is \cpoj-enriched, but in general it is not a \emph{small} category. 
There are two potential solutions to this problem. 

The simplest solution is to take \cat{K} as a suitable full subcategory of
$\mathcal{K}l(T)$ that meets our requirements. For example, if we are
interested in $T$-coalgebras whose base category is $\cat{Set}$ and
which have a carrier of cardinality below $\kappa$ then we can take 
\cat{K} to have exactly one set of cardinality $\lambda$ for any $\lambda<\kappa$. 
In particular, if $\kappa = \omega$ then \cat{K} is the dual category to the
Lawvere theory for $T$ \cite{hyland:power:2007}. This
category is a small full subcategory of $\mathcal{K}l(T)$.

The other possibility is to drop the smallness condition from
Definition~\ref{def:laxpresh},  and rewrite the whole
theory below so that the (not necessarily locally small) category
$\widetilde{\cat{K}}$ would fit it. Indeed, if \cat{K} is not small
then there is no guarantee that the hom-objects of
$\widetilde{\cat{K}}$ are sets, as they can form proper classes. In
other words, $\widetilde{\cat{K}}$ is not necessarily \pos- or
\cpo-enriched, as will be expected in the following section. 
Nonetheless, its hom-objects can be endowed with a partial order which
would turn them into partially ordered \emph{classes}. All of the
theorems presented in this paper would hold in this setting and the
definition of weak behavioural equivalence would be the same. However,
we think that this generality is not justified by the complexity of
dealing with classes instead of sets; hence, for the sake of
simplicity we decide to choose the first solution.
\end{remark}

Given two lax presheaves $F,G\in \widetilde{\cat{K}}$ and two oplax natural transformations  $\phi,\psi$ between them we define:
$$
\phi\leq \psi \iff \phi_Y(y) \leq \psi_Y(y) \text{ for any }y\in FY \text{ and any } Y\in \cat{K}.
$$
The partially ordered set $\widetilde{\cat{K}}(F,G)$ admits binary joins as $\phi \vee \psi$ is the family of morphisms $\{(\phi\vee \psi)_Y\colon FY \to GY\}_{Y\in \cat{K}}$ whose $Y$-component is given by:
$$
 (\phi\vee \psi)_Y(y) \defeq \phi_Y(y)\vee \psi_Y(y).
$$
It is straightforward to verify that $\phi\vee \psi$ is an oplax transformation from $F$ to $G$. In an analogous way we define $\bigvee_k \phi_k$ for an ascending $\omega$-chain $\phi_0\leq \phi_1\leq \ldots$ of oplax transformations. It is easy to see that $\bigvee_k \phi_k$ is also an oplax transformation between $F$ and $G$. Moreover, this type of supremum is preserved by the transformation composition. 

\begin{theorem}\label{theorem:K_tilde_op_LD}
The category $\widetilde{\cat{K}}$ is \cpoj-enriched and satisfies
\eqref{law:RD}. As a consequence, $\widetilde{\cat{K}}\op $ is also
\cpoj-enriched and satisfies \eqref{law:LD}.
\end{theorem}
\begin{proof}
Right distributivity of $\widetilde{\cat{K}}$ follows from the fact that the order on and the composition of oplax transformations is defined pointwise.
\end{proof}

For any object $X\in \mathsf{K}$ consider the functor $\widehat{X}\defeq\mathsf{K}(X,-)\colon \mathsf{K}\to \cpoj$ (called \emph{a representable presheaf}) which sends $Y$ to $\cat{K}(X,Y)$ and $f\colon Y\to Y'$ to the morphism $\cat{K}(X,f)\colon \cat{K}(X,Y)\to \cat{K}(X,Y')$
defined as $\cat{K}(X,f)(g) = f\circ g$. 
A family $ \psi = \{\psi_Y \colon \cat{K}(X,Y)\to \cat  K(X',Y)\}_{Y\in \cat{K}}$ is an oplax transformation between   $\widehat{X}$ and $\widehat{X'}$ if and only if for any $f\colon Y\to Y'$ and any $g\colon X\to Y$ we have: $$\psi_{Y'}(f\circ g)\leq f\circ \psi_Y(g).$$

\begin{example}
Let $f\colon X'\to X$ be a morphism in $\mathsf{K}$. Then the family $\widehat{f}$ of morphisms whose $Y$-component is given by $\mathsf{K}(X,Y)\to \mathsf{K}(X',Y); g\mapsto g\circ f$ is an oplax transformation between $\widehat{X}$ and $\widehat{X'}$.
\end{example}

Consider the functor $\widehat{(-)}\colon \cat{K} \to \widetilde{\cat{K}}\op $ which sends any object $X$ to $\widehat{X}$ and any morphism $i\colon X\to X'$ to the transformation $\widehat{i}$. Note that $\widehat{(-)}$ preserves the suprema of ascending $\omega$-chains (and hence, the order on hom-sets). The functor $\widehat{(-)}$ is a faithful functor from \cat{K} into a left distributive category $\widetilde{\cat{K}}\op $ akin to the Yoneda embedding of presheaf categories.

Let $X,Y$ be objects in $\cat K$. The functor $\widehat{(-)}\colon \cat{K}\to \widetilde{\cat{K}}\op $ restricted to hom-objects $\widehat{(-)}\colon \cat{K}(X,Y) \to  \widetilde{\cat{K}}\op (\widehat{X},\widehat{Y});g\mapsto \widehat{g}$ is a functor between the posets $\cat{K}(X,Y)$, $\widetilde{\cat{K}}\op (\widehat{X},\widehat{Y})$  viewed as categories. Moreover, the following result holds.

\begin{theorem}
	\label{thm:theta-left-adjoint}
The functor $\widehat{(-)}\colon \cat{K}(X,Y) \to  \widetilde{\cat{K}}\op (\widehat{X},\widehat{Y})$ admits a left adjoint $\Theta_{X,Y}$ defined for any oplax transformation $\phi\in \widetilde{\cat{K}}\op (\widehat{X},\widehat{Y}) = \widetilde{\cat{K}}(\widehat{Y},\widehat{X})$ by:
\[
\Theta_{X,Y}(\phi) = \phi_Y(id_Y)\text{.}
\]
\par
\end{theorem}
\begin{proof}
It is straightforward to show that for any $g\in \cat{K}(X,Y)$ and $\phi \in \widetilde{\cat{K}}\op (\widehat{X},\widehat{Y})$ we have:
$\Theta_{X,Y}(\phi) =  \phi_Y(id_Y) \leq g \iff \phi \leq \widehat{g}$. This ends the proof.
\end{proof}
We will often drop the subscript and write $\Theta$ instead of $\Theta_{X,Y}$ whenever the objects $X,Y$ can be deduced from the context. Thus, for any $X\in \cat{K}$ the functor $\widehat{X}\colon \cat{K}\to \cpoj$ together with the family \[
	\left\{\Theta\dashv \widehat{(-)}\colon \cat{K}(X,Y)\to \widetilde{\cat{K}}\op (\widehat{X},\widehat{Y})\right\}_{Y\in \cat{K}}
\] is a \cat{K}-adjunction.

\subsection{Free monads in order enriched categories}
\label{sec:free-monads-kop}

\begin{definition}
We say that an order enriched category  \emph{admits free monads} if for any endomorphism $\alpha\colon X\to X$ there is $\alpha^{\ast}\colon X\to X$ which is a monad such that:
\begin{itemize}
\item $\alpha\leq \alpha^{\ast}$,
\item if $\alpha\leq \sigma$ and $\sigma\colon X\to X$ is a monad  then $\alpha^{\ast}\leq \sigma$.
\end{itemize}
We will often call $\alpha^{\ast}$ the free monad over $\alpha\colon X\to X$.
\end{definition} 
\begin{proposition}\label{prop:cpoj_admits_saturation}
Any \cpoj-enriched category admits free monads.
\end{proposition}
\begin{proof}
For any $\alpha\colon X\to X$ put $\alpha^{\ast} = \bigvee_n (id_X \vee \alpha)^n$. It is easy to verify that $\alpha^{\ast}$ is the smallest monad satisfying $\alpha\leq \alpha^{\ast}$ and hence we leave the remainder of the proof to the reader.
\end{proof}

\paragraph{Free monads in $\widetilde{\cat{K}}\op $}\label{subsection:free_monads_in_K_tilde}
The purpose of this paragraph is to describe the transformation $\widehat{\alpha}^{\ast}:\widehat{X}\to \widehat{X}$ in $\widetilde{\cat{K}}\op $ for an arbitrary endomorphism $\alpha\colon X\to X$ in \cat{K}. Before we do that we will introduce some new notation. Given  $f\colon X\to Y$ in \cat{K} let $\alpha^{\ast}_f\colon X\to Y$ denote
the least solution to the equation:
\begin{equation}
	\tag{PS}
	\label{eq:PS}
	x = f \vee x \circ \alpha\text{.}
\end{equation}
 Such a solution exists in any \cpoj-enriched category $\cat K$ and is given by 
\[
	\alpha^{\ast}_f = \bigvee_{n < \omega} F^n(f)
	\text{,}
\]
where $F$ maps each $x \in \cat{K}(X,Y)$ to $f\vee x\circ \alpha \in \cat{K}(X,Y)$.
Intuitively, when $f$ is the identity, the above equation can be read at
each step as following $\alpha$ or as following $id_X$ making a ``self-loop''.
This extends to arbitrary $f \colon X \to Y$ where $Y$ is read as a quotient
of the state space. As exemplified by Section~\ref{sec:examples}, this
idea is at the core of saturation and weak bisimulations.

Theorem~\ref{thm:ps} below relates least solutions to \eqref{eq:PS} and free monads in $\widetilde{\cat{K}}\op$ thus providing a canonical characterization of it. Before we prove it we need one additional lemma.
\begin{lemma} \label{lemma:ps}
$
\alpha^{\ast}_f = \bigvee_{n < \omega} G^n(f)$,  where $G\colon \cat{K}(X,Y)\to \cat{K}(X,Y)$ is defined by $G(x) = x\vee   x\circ \alpha$.
\end{lemma}	
\begin{proof}
It is enough to prove that for any integer $n$ we have $G^n(f)=F^n(f)$. The assertion holds for $n = 0$. By induction assume $G^n(f)=F^n(f)$. We have:
\[
G^{n+1}(f) = G(G^{n}(f)) = F^n(f)\vee F^n(f) \circ \alpha \eqlabel[\text{i}]{eq:proof-lem-3.8-1} F^{n+1}(f).
\]
The equality \eqref{eq:proof-lem-3.8-1} follows also by induction. It is true for $n=0$. Assume that it holds for $n$, then we have:
\begin{align*}
&F^{n+2}(f) = f\vee F^{n+1}(f)\circ \alpha \eqlabel[\text{ii}]{eq:proof-lem-3.8-2}
f\vee F^n(f)\circ \alpha \vee  (F^{n}(f)\vee F^n(f)\circ \alpha)\circ \alpha \\=\;&F^{n+1}(f) \vee F^{n+1}(f)\circ \alpha
\text{.}
\end{align*}
where \eqref{eq:proof-lem-3.8-2} holds by the induction hypotheses.
\end{proof}
\begin{theorem}
	\label{thm:ps}
	$\widehat{\alpha}^{\ast}(f) = \alpha^{\ast}_f$.
\end{theorem}
\begin{proof}
By the guidelines of the proof of Proposition~\ref{prop:cpoj_admits_saturation} the monad $\widehat{\alpha}^{\ast}$ is given by $\bigvee_n (id_{\widehat{X}}\vee \widehat{\alpha})^n$. Since $(id \vee \widehat{\alpha})(x) = x\vee x\circ \alpha$, by Lemma~\ref{lemma:ps} we get the desired conclusion.
\end{proof}
Reworded, $\alpha^\ast_f$ is the image of $f$ under
the free monad on the embedded $\alpha$ in $\widetilde{\cat{K}}\op $. If we
think of $f \colon X \to Y$ as a quotient of the state space of $\alpha\colon X\to X$
then we may think of $\alpha^\ast_f$ as the saturation of $\alpha$ with respect to the partition induced by $f$ as we show in the next subsection and exemplify in Section~\ref{sec:examples}.

\subsection{Weak behavioural equivalence}
Following the approach of kernel bisimulations, we will say that a relation $R \rightrightarrows X$ is a ``weak behavioural equivalence'' for  a system $\alpha\colon X \to X \in \cat{K}$ if it is the kernel pair of some ``weak behavioural morphism'' for $\alpha$. Hence, the main step is to define a suitable notion of weak behavioural morphism.  This is the main aim of this subsection.

First, in Section~\ref{sec:first-glance-at-weak-beh} we recall Brengos' approach to saturation  \cite[§7]{brengos2013:corr}, which implicitly assumes left distributivity of the category $\cat{K} = \kl(T)$ to guarantee existence of a certain adjunction.  This approach is enough
to cover the cases where weak behavioural morphisms are refinements of saturated systems as for LTSs and Segala systems; however, there are many behaviours of interest that do not meet the left distributivity hypothesis, such as fully-probabilistic systems.  In order to circumvent this problem, in Section~\ref{sec:weak-beh-cpoj} we propose a construction for recovering \eqref{law:LD} by moving from \cat{K} to the left distributive category $\widetilde{\cat{K}}\op $ and performing saturation in the setting of presheaves and oplax transformations. As a preliminary step towards this construction, in Section~\ref{sec:saturation} we generalise  \cite[§7]{brengos2013:corr} to endomorphisms in \cat{K}.
Moreover, we provide two alternative characterizations of weak behavioural morphisms (and hence of weak behavioural equivalence): the first is
based on equation \eqref{eq:PS}, which subsumes the presentation of ``saturated transitions'' via recursive equations as in \cite{baier97:cav,mp2013:weak-arxiv,gp:icalp2014}; the second is specific to left distributive categories and captures the presentation of saturation systems via Milner’s double arrow construction.
Finally, in Section~\ref{sec:strongvsweak} we investigate which conditions
guarantee strong bisimulation to be also a weak behavioural equivalence.

\subsubsection{Weak behavioural morphisms: a first glance}\label{sec:first-glance-at-weak-beh}
In \cite[§7]{brengos2013:corr} Brengos introduces a categorical setting in which one can define final weak bisimilarity semantics for systems with internal moves. We recall basic ingredients of  this setting here. Assume $(T,\mu,\eta)$ is a monad on $\cat{C}$ whose Kleisli category $\kl (T) $ is order enriched. Define $\cat{C}_{T,\leq}$ to be the category consisting of all $T$-coalgebras as objects and maps $i\colon X\to Y$ between carrier of $\alpha\colon X\to TX$ and $\beta\colon Y\to TY$ satisfying the following condition as morphisms:
$$
	\begin{tikzpicture}[
		auto, scale=1, font=\small,
		baseline=(current bounding box.center)]	
		\matrix (m) [matrix of math nodes, row sep=4ex, column sep=5ex]{ 
			X & Y \\
			TX & TY \\
		};
		\draw[->] (m-1-1) to node {\(i\)} (m-1-2);
		\draw[->] (m-1-1) to node[swap] {\(\alpha\)} (m-2-1);
		\draw[->] (m-1-2) to node {\(\beta\)} (m-2-2);
		\draw[->] (m-2-1) to node[swap] {\(Ti\)} (m-2-2);
		\node at ($(m-1-1)!.5!(m-2-2)$) {\(\leq\)};
	\end{tikzpicture}
$$
Let $\cat{C}_{T,\leq}^\ast$ be the full subcategory of $\cat{C}_{T,\leq}$ consisting of coalgebras which are monads in $\mathcal{K}l(T)$ and assume the inclusion functor $\cat{C}_{T,\leq}^\ast\to \cat{C}_{T,\leq}$ admits a left adjoint $(-)^{\ast}\colon \cat{C}_{T,\leq} \to \cat{C}_{T,\leq}^\ast$ which is the identity on morphisms. It is easy to see that this condition implies that $\mathcal{K}l(T)$ admits free monads and Condition~\ref{condition:ordered_sat_monad} below is satisfied. Conversely, if $\mathcal{K}l(T)$ admits free monads and the implication below is true then the left adjoint which is the identity on morphisms exists and is given on objects by $\alpha\mapsto \alpha^\ast$.
\begin{equation}
	\label{condition:ordered_sat_monad}
	Ti\circ \alpha \leq \beta \circ i \implies Ti\circ \alpha^\ast \leq \beta^\ast \circ i 
\end{equation}

A natural consequence of this setting  is to define weak behavioural morphism for $\alpha\colon X\to TX$ to be a map $i\colon X\to Y$ in $\cat{C}$ for which there is $\beta\colon Y\to TY$ such that:
\[
	\begin{tikzpicture}[
		auto, scale=1, font=\small,
		baseline=(current bounding box.center)]	
		\matrix (m) [matrix of math nodes, row sep=4ex, column sep=5ex]{ 
			X & Y \\
			TX & TY \\
		};
		\draw[->] (m-1-1) to node {\(i\)} (m-1-2);
		\draw[->] (m-1-1) to node[swap] {\(\alpha^{\ast}\)} (m-2-1);
		\draw[->] (m-1-2) to node {\(\beta\)} (m-2-2);
		\draw[->] (m-2-1) to node[swap] {\(Ti\)} (m-2-2);
		\node at ($(m-1-1)!.5!(m-2-2)$) {\(=\)};
	\end{tikzpicture}
\]
The above equation can be stated in terms of the composition in $\mathcal{K}l(T)$ as $i^\sharp \circ \alpha^{\ast} = \beta\circ i^\sharp$. Reworded, a map is a weak behavioural morphism for $\alpha$ if it is a strong behavioural morphism for $\alpha^\ast$. This definition is expected to extend bisimulation, i.e., if a map is a strong behavioural morphism for $\alpha$ then it should also be a weak one. To guarantee this, \cite[§7]{brengos2013:corr} considers a second condition which is an analogue of Condition~\ref{condition:ordered_sat_monad} where all inequalities are replaced by equalities.  

It should be noted here that all examples this setting was tested on in \cite{brengos2013:corr} were examples for which $\mathcal{K}l(T)$ was \cpoj-enriched with \eqref{law:LD}. Left distributivity 
of the \cpoj-enriched category $\mathcal{K}l(T)$ implies Cond. (\ref{condition:ordered_sat_monad}) (see Example~\ref{example:endomorphisms_coalgebras} and Theorem~\ref{theorem:ordered_saturation_category_LD_RD} below for details). 
In the rest of this section we describe saturation and weak behavioural equivalence also in  
 \cpoj-enriched categories which are not necessarily \eqref{law:LD}. In this case, the saturation will be given on the level of
presheaves and oplax natural transformations on \cat{K} and not \cat{K} itself. 

\subsubsection{Saturation in left distributive categories of endomorphisms}\label{sec:saturation}
In this subsection we generalise the situation of $(-)^{\ast}\colon \cat{C}_{T,\leq} \to \cat{C}_{T,\leq}^\ast$ presented above, to categories of endomorphisms in \cat{K}.  Similarly to Section~\ref{sec:free-monads-kop}, where we showed that we can achieve \eqref{law:LD} by moving to $\widetilde{\cat{K}}\op$, this will allow us to achieve saturation even in categories which are not left distributive.

For the sake of exposition, in this subsection we assume \cat{K} to be an order enriched category which admits free monads (we will illustrate the full scenario of $\widetilde{\cat{K}}\op$ in Section~\ref{sec:weak-beh-cpoj}).
Assume $J$ is a subcategory of \cat{K}.  Since our prototypical example for \cat{K} is (a full subcategory of) $\mathcal{K}l(T)$ for a monad $(T,\mu,\eta)$ on $\cat{C}$ we should intuitively associate $J$ with (a full subcategory of) $\cat{C}$ together with the restriction of the inclusion functor $(-)^\sharp\colon \cat{C}\to \kl(T)$ to $J\to \cat{K}$. 

\begin{definition}
	\label{def:admit-saturation}
	We say that \cat{K}  \emph{admits saturation with respect to $J$} provided that for any $\alpha\colon X\to X$, $\beta\colon Y\to Y$ from \cat{K} and $f\colon X\to Y$ from $J$:
	\begin{align*}
		f\circ \alpha \mathrel{\leq }  \beta \circ f \implies f\circ \alpha^\ast \mathrel{\leq } \beta^\ast \circ f.  
	\end{align*}
	We say that \cat{K} \emph{admits saturation} if it admits saturation with respect to \cat{K}.
\end{definition}
The above definition can be equivalently restated in terms of existence of an adjunction satisfying some extra properties. Indeed, let $\cat{End}_J^\leq(\cat{K})$ be the category of all endomorphisms as objects and maps $i\colon X\to Y$ in $J$ between $\alpha\colon X\to X$ and $\beta\colon Y\to Y$ which satisfy $i\circ \alpha \mathrel{\leq} \beta \circ i$ as morphisms. Moreover, let  $\cat{End}_J^{\ast,\leq}(\cat{K})$ be the full subcategory of $\cat{End}_J^\leq (\cat{K})$ consisting only of endomorphisms which are monads as objects. The category \cat{K} admits saturation with respect to $J$ if and only if there is a left adjoint $(-)^\ast$ to the inclusion functor $\cat{End}_J^{\ast\leq}(\cat{K})\to \cat{End}_J^{\leq}(\cat{K})$ which is the identity on morphisms:
\[
	\begin{tikzpicture}[
			auto, scale=1.4,font=\small,
			baseline=(current bounding box.center)]	
		\matrix (m) [matrix of math nodes, row sep=4ex, column sep=6ex]{ 
			\cat{End}_J^{\leq}(\cat{K}) & \cat{End}_J^{\ast\leq}(\cat{K})\\
		};
		\draw[->, bend left] (m-1-1.north east) to node {\((-)^\ast\)} (m-1-2.north west);
		\draw[left hook->, bend left] (m-1-2.south west) to node {} (m-1-1.south east);
		\node[rotate=-90] at ($(m-1-1.east)!.5!(m-1-2.west)$) {\(\dashv\)};
	\end{tikzpicture}	
\]

\begin{example}\label{example:endomorphisms_coalgebras}
If $\cat{K}=\mathcal{K}l(T)$ and $J=\cat{C}$ for a monad $T$ on a category $\cat{C}$ then we have $\cat{End}_J^\leq(\cat{K}) = \cat{C}_{T,\leq}$ and $\cat{End}_J^{\ast,\leq}(\cat{K})=\cat{C}_{T,\leq}^\ast$ as in Section~\ref{sec:first-glance-at-weak-beh}.
\end{example}

\begin{theorem}\label{theorem:ordered_saturation_category_LD_RD}
If \cat{K} is \cpoj-enriched and left distributive then it admits saturation. \par
\end{theorem} 
\begin{proof}
	Assume $f\circ \alpha \leq \beta \circ f$. Hence, $f\circ \alpha \leq \beta \circ f \leq \beta^{\ast}\circ f$ and $f\leq \beta^{\ast}\circ f$. We prove $f\circ (id_X\vee \alpha )^n  \leq \beta^{\ast}\circ f$ by
	induction on $n$: for $n=0$ we have $f\circ (id_X\vee\alpha )^0  = f\leq \beta^{\ast}\circ f$ by the above and for $n+1$ we have:
	\begin{align*}
		& f\circ (id_X\vee \alpha)^{n+1}=f\circ (id_X\vee \alpha )^{n} \circ
		(id_X\vee\alpha )
		\oversetlabel[\text{i}]{eq:proof-ordered-saturation-1}{\leq}
		\beta^{\ast}\circ f \circ (id_X\vee \alpha )
		\oversetlabel[\text{ii}]{eq:proof-ordered-saturation-2}{\leq}
		\beta^{\ast}\circ (f\vee f\circ \alpha)\\
		\leq\;&\beta^\ast \circ
		\beta^\ast \circ f \leq \beta^{\ast}\circ f
		\text{,}
	\end{align*}
	where \eqref{eq:proof-ordered-saturation-1} holds by induction
	hypothesis and \eqref{eq:proof-ordered-saturation-2} by left
	distributivity. Since \cat{K} is \cpoj-enriched we have: $\textstyle f\circ \alpha^{\ast} = f\circ \bigvee_n (id_X\vee \alpha )^n =
	\bigvee_n f\circ (id_X\vee \alpha )^n \leq \beta^{\ast} \circ
	f\text{.}$
\end{proof}

\subsubsection{Weak behavioural morphisms for \cpoj-enriched categories}
\label{sec:weak-beh-cpoj}
In this subsection we give the definition of weak behavioural morphisms for endomorphisms in \cpoj-enriched categories. We provide three different characterizations: via saturation in $\widetilde{\cat{K}}\op$, via recursive equations, and (for left distributive categories only) via the ``double arrow construction''.

Henceforth, let \cat{K} be as in Sec.~\ref{subsection:the_category_K_l} and $J$ a subcategory of \cat{K} with all objects of \cat{K}. Differently from Sec.~\ref{sec:first-glance-at-weak-beh}, we do not assume \cat{K} to satisfy \eqref{law:LD}.

\paragraph{Weak behavioural morphisms via saturation in $\widetilde{\cat{K}}\op$}
To overcome the lack of left distributivity and retain the canonical construction offered
by saturation (as in Section~\ref{sec:first-glance-at-weak-beh}) we move along the embedding $\cat{K}\to \widetilde{\cat K}^{op}$ into a category admitting saturation, leading to the setting depicted in Section~\ref{sec:saturation}. 
\begin{proposition}\label{prop:tilde_K_saturation}
	The category $\widetilde{\cat{K}}\op $ admits saturation.
\end{proposition}
\begin{proof}
	The statement follows by Theorems~\ref{theorem:K_tilde_op_LD} and~\ref{theorem:ordered_saturation_category_LD_RD}.
\end{proof}

In the context of the above situation, the first step to define weak behavioural morphisms for $\alpha\colon X\to X$ in $\cat K$ is to consider the embedding
of $\alpha$ into $\widetilde{\cat{K}}\op$ via $\widehat{(-)}$. 
The second step is to saturate the endomorphism 
$\widehat{\alpha}\colon \widehat{X}\to \widehat{X}$ i.e., to consider its
free monad $\widehat{\alpha}^{\ast}$ in $\widetilde{\cat{K}}\op$ (along the lines of Section~\ref{sec:first-glance-at-weak-beh}). 
The final step is to return back to \cat{K} along the \cat{K}-adjunction via the projection $\Theta$. Then we can generalise the notion of weak behavioural morphism considered in the above subsections to allow us to move along $\Theta \dashv \widehat{(-)}$ and perform saturation in $\widetilde{\cat{K}}\op$.
\begin{definition}[Weak behavioural morphism---via saturation]
	\label{def:weak-behavioural-morphism}
	\label{def:weak-behavioural-morphism-ktilde}
		A map $i\colon X\to Y$ in $J$ is a \emph{weak behavioural morphism} for  an endomorphism $\alpha\colon X\to X$ in \cat{K} whenever there is an endomorphism $\beta\colon Y\to Y$ in \cat{K} such that:
		\begin{equation}
			\label{eq:WBM-ktilde}
			\tag{W-$\Theta$}
			\Theta(\widehat{i}\circ \widehat{\alpha}^\ast)
			=\Theta(\widehat{\beta} \circ \widehat{i})
			\text{.}
		\end{equation}
\end{definition}
Note that since each component of the composition $\Theta\circ \widehat{(-)}$ of the considered \cat{K}-adjunction  
acts as the identity on hom-objects of \cat{K} (cf.~Theorem~\ref{thm:theta-left-adjoint}), \eqref{eq:WBM-ktilde} becomes:
\[
	\Theta(\widehat{i}\circ \widehat{\alpha}^\ast)=
	\Theta(\widehat{\beta} \circ \widehat{i})=\Theta(\widehat{\beta\circ i}) =
	\beta\circ i
	\text{.}
\]
This formulation elicits how the definition builds on the notion of strong behavioural morphism and the previous subsections. In fact, \eqref{eq:WBM-ktilde} requires $i$ to be a map to a system $\beta$ refining the (projection of the) 
saturation of the system (embedded in $\widetilde{\cat{K}}\op$).
Intuitively, if we forget for a moment that saturation is performed in $\widetilde{\cat{K}}\op$ then,
the above is a strong behavioural morphism from the saturated system
$\alpha^*$ to the refinement $\beta$ i.e.~a weak behavioural morphism as in Section~\ref{sec:first-glance-at-weak-beh}.

We can now give the formal definition of weak behavioural equivalence.
\begin{definition}[Weak behavioural equivalence]
	\label{def:weak-behavioural-equivalence}
	A relation $R \rightrightarrows X$ in $J$ on the carrier of 
	an endomorphism $\alpha\colon X \to X \in \cat{K}$ is a \emph{weak behavioural equivalence}
	(or simply \emph{weak bisimulation}) for $\alpha$ 
	if and only if it is the 
	kernel pair of some weak behavioural morphism for $\alpha$ in $J$. 	
\end{definition}

\paragraph{Weak behavioural morphisms via recursive equations}
The connection between Definition~\ref{def:weak-behavioural-morphism-ktilde} and instances of weak bisimulation found in the literature may be not immediate. 
To this end, we provide an equivalent characterisation based on least solutions to \eqref{eq:PS}.
In fact,
as we illustrate in Section~\ref{sec:examples},
weak bisimulations (e.g.,~\cite{baier97:cav,mp2013:weak-arxiv,gp:icalp2014}) 
are based on recursive equations subsumed\footnote{%
	Weak bisimulations defined using Milner's double arrow construction have been already covered by Brengos' results (cf.~Section~\ref{sec:first-glance-at-weak-beh});  in Theorem~\ref{thm:double-arrow} we will show that this construction is a special case of the definitions provided in this section.} by \eqref{eq:PS}.

\begin{theorem}\label{thm:wbm-via-ps}
	A map $i\colon X \to Y$ in $J$ is a weak behavioural morphism for $\alpha$ if and only if there is  $\beta\colon Y\to Y$ in \cat{K} satisfying $\alpha^\ast_i = \beta \circ i$.
\end{theorem}
\begin{proof}
	In order to prove the statement it is enough to show that $\Theta(\widehat{i}\circ \widehat{\alpha}^\ast) = \alpha^\ast_i$. Indeed, we have the following:
	\begin{align*}
	& \textstyle\Theta(\widehat{i}\circ \widehat{\alpha}^\ast) = \Theta (\widehat{i}\circ \bigvee_n (id \vee \widehat{\alpha})^n) = \Theta ( \bigvee_n \widehat{i}\circ (id \vee \widehat{\alpha})^n)
	 \eqlabel[\text{i}]{eq:proof-thm-wbm-1}\bigvee_n \Theta( \widehat{i}\circ (id \vee \widehat{\alpha})^n)
	 \\ 
	 \eqlabel[\text{ii}]{eq:proof-thm-wbm-2}\;& \textstyle\bigvee_n G^n(i) \eqlabel[\text{iii}]{eq:proof-thm-wbm-3} \alpha^\ast_i\text{.} 
	\end{align*}
	In the above, $G:\cat K(X,Y)\to \cat  K(X,Y); x\mapsto x\vee x\circ \alpha$ and the equality \eqref{eq:proof-thm-wbm-1} follows by the fact that $\Theta$ is a left adjoint and, therefore, it preserves arbitrary suprema. The identity \eqref{eq:proof-thm-wbm-2} follows by left distributivity of $\widetilde{\cat K}^{op}$ and induction. Finally, the equality \eqref{eq:proof-thm-wbm-3} is a consequence of Lemma~\ref{lemma:ps}.
\end{proof}

Therefore, we can rephrase Definition~\ref{def:weak-behavioural-morphism-ktilde} 
to characterise weak behavioural morphism by means of least solutions to the recursive equation \eqref{eq:PS}.

\begin{definition}[Weak behavioural morphism---via recursive equations]
	\label{def:weak-behavioural-morphism-ps}
	A morphism $i\colon X\to Y$ in $J$ is a \emph{weak behavioural morphism} for an endomorphism $\alpha\colon X\to X$ in \cat{K} whenever there is an endomorphism $\beta\colon Y\to Y$ in 
	\cat{K} such that:
	\begin{equation}
		\label{eq:WBM-PS}
		\tag{W-PS}
		\alpha^\ast_i = \beta \circ i
		\text{.}
	\end{equation}
\end{definition}

\begin{remark}
Although they determine the very same class of weak behavioural morphisms, Definitions~\ref{def:weak-behavioural-morphism-ktilde} and \ref{def:weak-behavioural-morphism-ps} unveil different aspects of the same concept: the former shows how these morphisms arise in a canonical way from the concept of saturation and free monads, while the latter offers a more immediate connection with instances of weak bisimulations since \eqref{eq:PS} subsumes the recursive equations used in concrete settings such as \cite{baier97:cav,mp2013:weak-arxiv,gp:icalp2014} (cf.~Section~\ref{sec:examples}).
\end{remark}

\paragraph{Other embeddings to left distributive categories}

Although the setting presented in the previous subsection is to some extent canonical, the reader may wonder what is so special about the category $\widetilde{\cat{K}}\op $, besides for being \eqref{law:LD} and having the  family of $\cat{K}$-adjunctions. The answer to this question is: \emph{nothing}. It turns out that we may move the process of saturation to any other left distributive category and obtain the very \emph{same} notion of behavioural morphism. Indeed, we have the following: take $\cat{K}'$ to be a \cpoj-enriched category that satisfies \eqref{law:LD} and let $\widehat{(-)}'\colon \cat{K}\to \cat{K}'$ be a faithful functor which preserves the order. Moreover, assume for that for any $X\in \cat{K}$ we are given a \cat{K}-adjunction consisting of  the functor $\widehat{X}\colon \cat K\to \cpoj$ and the family 
\[
	\left\{\Theta'\dashv \widehat{(-)}'\colon \cat{K}(X,Y)\to \cat{K}'(\widehat{X}',\widehat{Y}')\right\}_{Y\in \cat{K}}
\]
for which $\Theta'\circ \widehat{(-)}'$ is the identity on $\cat{K}(X,Y)$. 
It is straightforward to check that the proof of Theorem~\ref{thm:wbm-via-ps}  can be directly generalized to this setting. In other words, a morphism $i\colon X\to Y$ satisfies an analogous condition to \eqref{eq:WBM-ktilde} for $\alpha\colon X\to X$ in which we replace $\widehat{(-)}$ and $\Theta$ with $\widehat{(-)}'$ and $\Theta'$, respectively, iff there is $\beta\colon Y\to Y$ s.t.~$\alpha_i^\ast = \beta \circ i$.  This yields the following result.

\begin{theorem}
	\label{thm:saturation-independent}
	For any $\alpha\colon X\to X$ and $i\colon X\to Y$ in $\cat{K}$ we have:
	\[
		\Theta(\widehat{i}\circ \widehat{\alpha}^\ast) = \alpha^\ast_i = \Theta'(\widehat{i}'\circ \widehat{\alpha}'^\ast) \text{.}
	\]
	\par
\end{theorem}
\begin{proof}
	The first equivalence follows from Theorem~\ref{thm:wbm-via-ps}.  We prove the second equivalence.
	Since $\cat{K}'$ is \cpoj-enriched we have $\widehat{\alpha}'^\ast = \bigvee_{n<\omega} (id\vee \widehat{\alpha}')^n$. Hence, 
	\[\textstyle
		\widehat{i}'\circ \widehat{\alpha}'^\ast = \widehat{i}'\circ \bigvee_{n<\omega} (id\vee \widehat{\alpha}')^n = \bigvee_{n<\omega} \widehat{i}'\circ (id\vee \widehat{\alpha}')^n=\bigvee_{n<\omega} G'^n(\widehat{i}')\text{,}
	\]
	where the last equality follows by \eqref{law:LD} and induction. Here, $G'(x) = x\vee x\circ \widehat{\alpha}'$ is a map on $\cat{K}'(\widehat{X}',\widehat{Y}')$. Since $\Theta'$ is a left adjoint it preserves arbitrary suprema. Therefore, 
	\[\textstyle
	\Theta'(\widehat{i}'\circ \widehat{\alpha}'^\ast) = \Theta'(\bigvee_{n<\omega} G'^n(\widehat{i}')) = \bigvee_{n<\omega} \Theta'(G'^n(\widehat{i}')) = \bigvee_{n<\omega} G^n(i)\text{,}
	\]
	where $G\colon \cat{K}(X,Y)\to \cat{K}(X,Y)$ is given
	by the assignment $x\mapsto x\vee x\circ \alpha$. This proves the assertion as $\bigvee_n G^n(i) = \alpha^\ast_i$ (see Lemma~\ref{lemma:ps}).
\end{proof}

Thus, Definition~\ref{def:weak-behavioural-morphism-ktilde} can be rephrased to define weak behavioural morphisms via saturation along any \cat{K}-adjunction embedding \cat{K} to some left distributive category---even \cat{K} itself.

\paragraph{Weak behavioural morphisms via double arrow construction}
In fact, if \cat{K} is left distributive weak behavioural equivalences for $\alpha$ 
are exactly strong ones for $\alpha^\ast$; this corresponds precisely
to the well-known situation of weak bisimulations for LTS 
or Segala systems. This approach is often called Milner's 
\emph{double arrow construction}
(as transitions of the saturated system are usually depicted 
using double arrows---cf.~Section~\ref{sec:examples}). 

\begin{theorem}\label{thm:double-arrow}
	If \cat{K} is left distributive then $i$ is a weak behavioural morphism if and only if there is  $\beta\colon Y\to Y$ in \cat{K} such that:
	$i\circ \alpha^\ast = \beta \circ i$.
\end{theorem}
\begin{proof}
	To prove the above statement we can adopt either
	Definition~\ref{def:weak-behavioural-morphism-ktilde}
	or Definition~\ref{def:weak-behavioural-morphism-ps}
	as a characterisation of weak behavioural morphisms.
	In the former case, the proof of this statement follows easily by Theorem~\ref{thm:saturation-independent}. Indeed, since $\cat K$ is assumed to be \eqref{law:LD} we put $\cat K' = \cat K$ and $\widehat{(-)}', \Theta'$ to be identities.
	In the later case, the proof of the statement follows routinely
	by noting that $G(x) = x \vee x \circ \alpha = x \circ (id_X \vee \alpha)$. By Theorem~\ref{thm:wbm-via-ps} either is sufficient.
\end{proof}

Thus, in presence of left distributivity saturation is independent from
the choice of maps in $J$. In this situation we can use the following
characterisation of weak behavioural morphisms.
\begin{definition}[Weak behavioural morphism---via double arrow]
	\label{def:weak-behavioural-morphism-double}
	Let \cat{K} be left distributive.
	A map $i\colon X\to Y$ in $J$ is a \emph{weak behavioural morphism} for  $\alpha\colon X\to X$ in \cat{K} whenever there is an endomorphism $\beta\colon Y\to Y$ in 
	\cat{K} such that:
	\begin{equation}
		\label{eq:WBM-DA}
		\tag{W-DA}
		i \circ \alpha^\ast = \beta \circ i
		\text{.}
	\end{equation}
\end{definition}

\subsubsection{Strong behavioural morphisms as weak behavioural morphisms}\label{sec:strongvsweak}
Usually, strong bisimulations found in the literature are also weak bisimulations. In this subsection we show that this holds at the level of generality of our treatment, when the following weaker form of right distributivity holds true.
\begin{definition}[$J$-right distributivity]
 We say that \cat{K} is \emph{$J$-right distributive} when, for all $f,g$ in \cat{K} and 
	$h$ in $J$ with suitable domain and codomain:
	\begin{equation}
	\tag{$J$-RD}\label{law:JRD}
	(f\vee g)\circ h = f\circ h \vee g\circ h.
	\end{equation}
\end{definition}

Under this mild assumption any strong behavioural morphism
is a weak behavioural morphism; thus any strong bisimulation is a weak bisimulation.
\begin{theorem}
	\label{thm:weak-is-complete-wrt-strong}
	Let \cat{K} satisfy \eqref{law:JRD} and assume $\alpha\colon X\to X$ is an endomorphism in \cat{K}. If $i\colon X\to Y\in J$ is a strong behavioural morphism for $\alpha$ then $i$ is a weak behavioural morphism for $\alpha$.
\end{theorem}
\begin{proof}
Assume $i\circ \alpha = \beta \circ i$ for some $\beta:Y\to Y\in \cat K$. We will now prove $\alpha_i^\ast = \beta^\ast_{id_Y} \circ i$. Consider the maps $F\colon \cat{K}(X,Y)\to \cat{K}(X,Y); x\mapsto i \vee x\circ \alpha$ and $H\colon \cat{K}(Y,Y)\to \cat{K}(Y,Y);x\mapsto id_Y\vee x\circ \beta$. We have $F(i) = i\vee i\circ \alpha = i \vee  \beta\circ i = (id\vee \beta)\circ i=H(id_Y)\circ i$. By induction assume that $F^n(i)= H^n(id_Y) \circ i$. For $n+1$ we have:
\begin{align*}
&F^{n+1}(i)= i\vee F^n(i)\circ \alpha = i\vee H^n(id_Y) \circ i\circ \alpha =
i\vee H^n(id_Y) \circ \beta \circ i 
\\=\;&\left( id \vee H^n(id_Y) \circ \beta\right) \circ i = H^{n+1}(id_Y)\circ i.
\end{align*}
This proves the claim since \cat{K} is \cpoj-enriched and $\alpha^\ast_i = \bigvee F^n(i)$, $\beta^\ast_{id_Y} = \bigvee H^n(id_Y)$.
\end{proof}

Referring to our prototypical setting, this weaker form of right distributivity holds for all examples of $\kl(T)$ and $\cat{C}$ in this paper. 
In fact, if binary joins are a pointwise extension of data from the monad $T$ 
then $\cat{C}$-right distributivity holds.

\begin{definition}
	\label{def:pointwise-join}
	\looseness=-1
	Let $\cat{C}$ have binary products.
	We say that binary joins in $\kl(T)$ are \emph{pointwise}
	if there is a natural transformation $\ovee \colon T \times T \To T$
	such that for any $f,g \colon X \to Y \in \kl(T)$ the following equality holds:
	\[
		f \vee g = \ovee_Y \circ \langle f, g\rangle
	\]
	where ${\langle f, g\rangle \colon X \to TY \times TY}$ is given
	by the universal property of products in \cat{C}.	
\end{definition}
\begin{lemma}
	\label{lem:pointwise-jrd}
	If $\kl(T)$ has pointwise joins then
	it is \cat{C}-right distributive.
\end{lemma}
%\begin{proofsketch}
%	$(f\vee g)\bullet j^\sharp = \ovee_Z \circ \langle f, g\rangle \circ j = 
%	\ovee_Z \circ \langle f \circ j, g \circ j\rangle =
%	f\bullet h \vee g \bullet j^\sharp$.
%\end{proofsketch}
\begin{proof} %atend}
	Let $f,g \in \kl(T)$, $j \in \cat{C}$ with suitable domains and codomains,
	$\ovee$ as in Definition~\ref{def:pointwise-join}, and
	$\bullet$ and $\circ$ be composition in $\kl(T)$ and \cat{C}.
	Clearly $f \bullet j^\sharp = f \circ j$ and thus
	$
		(f\vee g)\bullet j^\sharp = 
		\ovee_Z \circ \langle f, g\rangle \circ j = 
		\ovee_Z \circ \langle f \circ j, g \circ j\rangle =
		f\bullet j^\sharp \vee g \bullet j^\sharp
		\text{.}
	$
\end{proof} %atend}

\section{Applications}\label{sec:examples}
In this section we illustrate the generality of the results presented
by listing some representative examples of behaviours fitting our framework;
these examples cover seamlessly quantitative aspects, dynamically allocated resources, and continuous state systems among others.

For the sake of exposition we assume \cat{K} and ${J}$ to
be $\kl(T)$ and $\cat{C}$ respectively, bearing in mind that cardinality
can be handled as described in Section~\ref{subsection:the_category_K_l}
and that \eqref{eq:PS} admits a least solution in any \cpoj-enriched
category.

All the examples we consider in this paper are set 
in concrete categories and the enrichment stems from 
the behavioural functor itself. In particular, binary 
joins are pointwise (cf.~Defintion~\ref{def:pointwise-join})
and, by Lemma~\ref{lem:pointwise-jrd}, right distribute over 
composition with morphisms coming from \cat{C}, 
i.e.~all the examples are set in 
\cat{C}-right distributive \cpoj-enriched Kleisli categories.

\subsection{Labelled transition systems}
\label{sec:ex-lts}

\paragraph{Powerset monad}
A standard example of a monad that fits our setting is the powerset
monad $\mathcal{P}\colon \Set\to\Set$, whose Kleisli category
$\kl(\mathcal{P})$ is isomorphic to $\mathsf{Rel}$, the category of
sets as objects, and binary relations as morphisms with relation
composition as the morphism composition. It is easy to see that
hom-sets of $\kl(\mathcal{P})$ admit arbitrary joins which are
preserved by the composition. Hence, this category is enriched over
\cpoj and both \eqref{law:LD} and \eqref{law:RD} hold. 

\paragraph{Weak behavioural equivalence}
Labelled transition systems  with silent actions \cite{milner:cc,sangiorgi2011:bis} are modelled as coalgebras for the endofunctor $\mathcal{P}(A_\tau \times \Id)$ \cite{rutten:universal}.  In Example~\ref{example:two_monads_on_Kl_P} we have seen how the functor $\mathcal{P}(A_\tau\times \Id)$ can be equipped with a monadic structure which handles unobservable moves internally. 

\begin{proposition}
	\label{prop:lts-cpoj}
	The category $\kl(\mathcal{P}(A_\tau\times \Id))$ is \cpoj-enriched
	and satisfies both left and right distributivity.
\end{proposition}
\begin{proof} %atend}
	It follows by Theorem~\ref{thm:t_kleisli_s_parameterised_saturation}.
	Indeed, $\kl(\mathcal{P})\cong \cat{Rel}$ is an \cpo-enriched
	category with binary joins. Since the powerset monad $\mathcal{P}$
	is commutative, the polynomial functor $S=A_\tau \times \Id$
	lifts to $\overline{A_\tau \times \Id}\colon \kl(\mathcal{P})\to
	\kl(\mathcal{P})$. Moreover, $\overline{A_\tau \times \Id}$ can be equipped
	with a monadic structure as described in Example~\ref{example:two_monads_on_Kl_P}.
	Since $\overline{A_\tau \times \Id}$ preserves
	joins in $\kl(\mathcal{P})$ \cite{brengos2013:corr} it is
	locally continuous. This proves that $\kl(\overline{A_\tau \times \Id}) =
	\kl(\mathcal{P}(A_\tau\times \Id))$ is
	\cpo-enriched. Additionally, it is easy to see that hom-sets in
	$\kl(\mathcal{P}(A_\tau\times \Id))$ admit binary
	joins and are left- and right distribute over composition. 
	This proves the assertion.
\end{proof} %atend}
\begin{proposition}
	For $\mathcal{P}(A_\tau \times \Id)$-coalgebras,
	weak behavioural equivalence corresponds to Milner's weak bisimulation
	for labelled transition systems \cite{milner:cc,sangiorgi2011:bis}.
	
\end{proposition}
\begin{proof} %atend}
	Thanks to Theorems~\ref{thm:wbm-via-ps}, \ref{thm:double-arrow}, \ref{theorem:ordered_saturation_category_LD_RD}
	weak behavioural equivalences (Definition~\ref{def:weak-behavioural-equivalence})
	correspond to strong bisimulations on saturated systems
	(cf.~Section~\ref{sec:first-glance-at-weak-beh}).
	We conclude by the adequacy of saturation as in \cite{brengos2013:corr}.
\end{proof} %atend}
Although we considered the unlimited powerset, the result holds true also for the limited (but transfinite) case
since LTS are a special case of weighted transition systems Section~\ref{sec:ex-weighted}. The finite powerset does
not meet the assumptions of the framework but this is coherent
to the well-known fact that saturation of image-finite LTS
may yield LTS that are not image-finite, but image-countable;
for a concrete example consider the LTS
$(\mathbb{N},A_\tau,\{(n,\tau,n+1)\mid n \in \mathbb{N}\})$
and its saturation
$(\mathbb{N},A_\tau,\{(n,\tau,n+m)\mid n, m \in \mathbb{N}\})$.

\subsection{Segala systems}
\label{sec:ex-segala}
Probabilistic systems \cite{sl:njc95}, known in the coalgebraic
literature as ``possibilistic and probabilistic systems'' or, more often, as \emph{Segala systems}, can be modelled as coalgebras for the functor $\mathcal{P D}(A_\tau \times \Id)$
\cite{sokolova11}. However, due to a lack of a distributive law
between monads $\mathcal{P}$ and $\mathcal{D}$ the part of the above
functor modelling the nature of branching in Segala systems
fails to be a proper monad \cite{jacobs08:cmcs,varaccawinskel2006:mscs}.

\paragraph{Convex combinations monad(s)}
The convex combinations monad was first introduced in its full
generality by Jacobs in \cite{jacobs08:cmcs} to study trace semantics
for combined possibilistic and probabilistic systems. Independently,
Brengos \cite{brengos2013:corr} and Goncharov and Pattison
\cite{gp:icalp2014} have tweaked Jacobs' construction slightly, so
that the resulting monads are more suitable to model 
Segala systems and their weak bisimulations.  
Jacobs' monad, Brengos' monad and
Goncharov-Pattison's monad form Kleisli categories
which are \cpo-enriched and whose hom-sets admit binary
joins. For the purposes of this paper we take the convex
combinations monad $\CM\colon \Set\to \Set$ to be
that considered in \cite[§8]{brengos2013:corr}. For the sake of completeness of this paper we recall the definition of $\CM $ here.

By $[0,\infty)$ we denote the semiring $([0,\infty),+,0,\cdot,1)$ of non-negative real numbers with ordinary addition and multiplication. By a $[0,\infty)$-\emph{semimodule} we mean a commutative monoid with actions $[0,\infty)\times (-)\to (-)$ satisfying axioms listed in e.g.~\cite{golan:99}. For a set $X$ and a function $f\colon X\to Y$ define
\[
	\mathcal{M} X \defeq \{\phi\colon X\to [0,\infty) \mid \{x \mid \phi(x) \neq 0\} \text{ is finite}\} \qquad
	\mathcal{M} f(\phi)(y) \defeq \textstyle \sum_{x\in f^{-1}(y)}\phi(x)
	\text{.}
\]
We will denote elements $\phi \in \mathcal{M}X$ using the formal sum notation by $\sum_x \phi(x) \cdot x$ or simply by $\sum_{i=1}^n \phi(x_i)\cdot x_i$ if $\{x \mid \phi(x) \neq 0\} = \{x_1,\ldots,x_n\}$. The set $\mathcal{M} X$ carries a monoid structure via pointwise operation of addition, and $[0,\infty)$-action via
$(a\cdot \phi) (x) \defeq a\cdot \phi(x),$
which turn $\mathcal{M} X$ into a free semimodule over $X$ (see e.g.~\cite{brengos2013:corr, golan:99,jacobs08:cmcs} for details). For a non-empty subset $U\subseteq \mathcal{M} X$ we define its convex closure by:
\[
	\overline{U} \defeq 
	\left\{ a_1\cdot \phi_1 +\ldots +a_n\cdot \phi_n 
	\;\middle|\; 
	\phi_i\in U, a_i\in [0,\infty) \text{ s.t.~} \textstyle\sum_{i=1}^na_i  =1
	\right\}
	\text{.}
\]
We call a subset $U\subseteq \mathcal{M}X$ \emph{convex} if $\overline{U} = U$. Finally, for a set $X$ and a function $f\colon X\to Y$ we define:
\[
 \CM X \defeq \{U\subseteq \mathcal{M}X \mid U \text{ is convex and non-empty}\} \qquad
 \CM f(U) \defeq {\mathcal{M}f(U)}\text{.}
\]
The assignment $\CM\colon\Set\to \Set$ is a functor which carries a monadic structure \cite{brengos2013:corr}.
Its multiplication and unit are given on their components by:
\[
	\mu_X(U) \defeq \bigcup_{\phi \in U} \sum_{V\in \CM X}\{\phi(V)\cdot \psi \mid \psi \in V\}
	\qquad
	\eta_X(x) \defeq \{1\cdot x\}
	\text{.} 
\]

\paragraph{Weak behavioural equivalence}
In \cite[§8]{brengos2013:corr} Brengos proposes to
consider Segala systems as coalgebras for the endofunctor $\CM (A_\tau\times
\Id)$ 
(see \cite{brengos2013:corr} for a discussion on consequences of this treatment). From now on, whenever we refer to ``Segala
systems'' we refer to $\CM (A_\tau\times
\Id)$-coalgebras. The functor
$\CM (A_\tau\times \Id)$ carries a monadic structure
as described in \cite{brengos2013:corr}. Moreover, we have the following.
\begin{proposition}
	\label{prop:segala-cpoj}
	The category
	$\kl(\CM (A_\tau\times \Id))$ is enriched over \cpoj
	and satisfies left distributivity.
\end{proposition}
\begin{proof} %atend}
	The proof follows the lines of the proof of 
	Theorem~\ref{prop:lts-cpoj} for $\CM $ instead of $\mathcal{P}$. 
	Indeed, by \cite[Lem.~8.6]{brengos2013:corr} the category $\kl(\CM )$ is 
	\cpo-enriched and admits binary joins. 
	The functor $A_\tau\times \Id$ lifts to an endofunctor $\overline{A_\tau 
	\times \Id}$ on $\kl(\CM )$ which can be equipped with a monadic 
	structure following the guidelines of \cite{brengos2013:corr}. The functor 
	$\overline{A_\tau \times \Id}$ is locally continuous \cite{brengos2013:corr}. By left distributivity of $\kl(\CM )$ \cite[Lem.~8.6]{brengos2013:corr} and Theorem~\ref{thm:t_kleisli_s_parameterised_saturation} we conclude that $\kl(\CM (A_\tau \times \Id))$ is left distributive \cpoj-enriched category.
\end{proof} %atend}

Segala systems just like LTSs, are endomorphisms in a \cpoj-enriched 
category satisfying \eqref{law:LD}  allowing us
to reduce the correspondence with weak behavioural equivalence to the
context of \cite{brengos2013:corr} by applying
Theorem~\ref{theorem:ordered_saturation_category_LD_RD} and~\ref{thm:double-arrow}.

\begin{proposition}
	For $\CM (A_\tau \times \Id)$-coalgebras,
	weak behavioural equivalence corresponds to Segala's weak 
	(convex) bisimulation \cite{brengos2013:corr}.
\end{proposition}
\begin{proof} %atend}
	Thanks to Theorems~\ref{theorem:ordered_saturation_category_LD_RD}, \ref{thm:wbm-via-ps} and~\ref{thm:double-arrow} 
	weak behavioural equivalence for $\alpha$
	corresponds to strong bisimulations on $\alpha^{\ast}$.
	Hence, it coincides with weak (convex) bisimulation on $\alpha$ \cite{brengos2013:corr}.
\end{proof} %atend}

\subsection{Weighted systems}
\label{sec:ex-weighted}
Weighted systems are transition systems whose transitions are 
given weights drawn from a semiring where addition and multiplication
characterise transition branching and chaining thus offering a modular
approach to model several quantitative aspects as shown e.g.~in 
\cite{ks2013:w-s-gsos,mp:2014ultras-gsos}

Let us recall some preliminary definitions.
A semiring $(W,+,0,\cdot,1)$ is said to be 
\emph{positively ordered} whenever its carrier $W$ admits a partial
order $(W,\leq)$ such that the unit $0$ is the bottom element of this ordering
and semiring operations are monotonic in both components i.e.: $x \leq y$ 
implies $x \mathrel{\diamond} z \leq y \mathrel{\diamond} z$ and 
$z \mathrel{\diamond} x \leq z \mathrel{\diamond} y$ for $\diamond \in 
\{+,\cdot\}$ and $x,y,z\in W$.  
A semiring is positively ordered if and only if it is \emph{zerosumfree} i.e.~$x +
y = 0$ implies $x = y = 0$; the natural order $x \lhd y
\iff \exists z. x + z = y$ is the weakest one rendering $W$ positively
ordered.

A positively ordered semiring is said to be \emph{$\omega$-complete} if it has
countable sums given as
$\sum_{i \leq \omega} x_i = \sup\{
		\sum_{j\in J} x_j \mid
		J \subset \omega \}
$. It is called \emph{$\omega$-continuous} if suprema of ascending
$\omega$-chains exist and are preserved by both operations
i.e.:
$y \mathrel{\diamond} \bigvee_{i} x_i = \bigvee_{i} y \mathrel{\diamond} x_i$
and
$\bigvee_{i} x_i \mathrel{\diamond} z	= \bigvee_{i} x_i \mathrel{\diamond} z$
for $\diamond \in \{+,\cdot\}$ and $x,y,z\in W$.
Examples of such semirings are:
the boolean semiring, the arithmetic semiring of non-negative real numbers with infinity and the tropical semiring.

Henceforth we assume $(W,+,0,\cdot,1)$ to be an $\omega$-complete and 
$\omega$-continuous semiring and denote it by its carrier, when confusion 
seems unlikely.

\paragraph{Countable generalised multiset monad}
Consider the \Set endofunctor $\mathcal F_W$
given on any set $X$ and on any function $f \colon X \to Y$ as
follows:
\[
	\mathcal F_WX \defeq
	\{\phi \colon X \to W \mid |\{x \mid \phi(x) \neq 0\}| \leq \omega \}
	\qquad
	\mathcal F_Wf(\phi)(y) \defeq \sum_{x \in f^{-1}(y)\}} \phi(x) \text{.}
\]
This functor extends to the \emph{countable generalised multiset monad}  
(a.k.a.~\emph{$\omega$-complete semimodule monad}) whose multiplication
$\mu$ and unit $\eta$ are given on each component by:
\[
	\mu_X(\phi)(x) \defeq \sum_{\psi \in \mathcal F_WX} \phi(\psi) \cdot \psi(x)	
	\qquad
	\eta_X(x)(x') \defeq
	\begin{cases}
		1 & \text{ if $x = x'$} \\
		0 & \text{ otherwise.}
	\end{cases}
\]
For any $X,Y \in \Set$ and $f,g\in \kl(\mathcal{F}_W)(X,Y)$ define:
\[
	f \leq g \defiff  f(x)(y) \leq g(x)(y) \text{ for any } x\in X,\ y \in Y.
\]
Likewise, suprema of $\omega$-ascending chains and (if present) binary joins are
extended from $W$ to $\mathcal{F}_W$ and eventually to $\kl(\mathcal{F}_W)$.
\begin{proposition}
	\label{prop:kl-fw-cpoj}
	Let $W$ have binary joins.
	The category $\kl(\mathcal F_W)$ is \cpoj-enriched.
	\par 
\end{proposition}
\begin{proof}[Proof]
	Follows from $W$ being $\omega$-continuous and joins being pointwise.
\end{proof}
The countable powerset monad $\mathcal P_{\omega}$ is
$\mathcal F_\mathbb B$ where $\mathbb B$ denotes the
boolean semiring $(\{\mathtt{t\!t}, \mathtt{f\!f}\},\lor,\land)$.
The monad of discrete probability distributions (with countable support) $\mathcal D$ 
is a submonad of $\mathcal F_{[0,\infty]}$ for the semiring of
non-negative real numbers extended with the infinity
(i.e.~the free $\omega$-completion of the partial semiring $[0,1]$).
Note that \eqref{law:LD} is satisfied in the first case
but not in the latter; in fact, for any
$\phi_1 \leq \psi_1, \phi_2 \leq \psi_2 \in (0,\infty)$ (e.g.~$1,2,3,4$)
$(\phi_1 \vee \psi_1) + (\phi_2 \vee \psi_2)
\neq (\phi_1+\psi_1) \vee (\phi_2+\psi_2)$ providing a counter example to \eqref{law:LD} in $\kl(\mathcal F_{[0,\infty]})$.

\paragraph{Weighted transition systems}
The monad $(\mathcal F_W,\mu,\eta)$ is commutative and its 
double strength is defined, on each $X,Y$,
as:
\[
	\dstr_{X,Y}(\phi,\psi)(x,y) = \phi(x)\cdot\psi(y)
\]
generalising the case of the probability distribution monad $\mathcal{D}$.
Therefore, polynomial functors have canonical liftings to $\kl(\mathcal F_W)$.
\begin{lemma}
	\label{lem:kl-fw-polynomial-continuous}
	Canonical liftings of polynomial functors
	are locally continuous.\par
\end{lemma}
%\begin{proofsketch}
%	Follows by an argument similar to \cite[Lem.~2.6]{hasuo07:trace}.
%\end{proofsketch}
\begin{proof} %atend}
	The proof is carried out by structural recursion akin to
	Hasuo's proof of the lemma for $\mathcal P$ and $\mathcal D$ 
	(cf.~\cite[Lem.~2.6]{hasuo07:trace}). Because the ordering
	is pointwise it suffices to show that $\mathsf{dstr}$ and coprojections
	are continuous map between CPOs.
	Both follow by $W$ being continuous	and suprema being pointwise.
\end{proof} %atend}

Zero morphisms in $\kl(\mathcal F_W)$ map every element of
their domain to the zero function $x \mapsto 0$.
Following the guidelines of Theorem~\ref{thm:monadic-fg-zero-morph} we impose a monadic structure on $\mathcal F_W(F + \Id)$
such that its Kleisli category is \cpoj-enriched.
\begin{proposition}
	Let $F \Set \to \Set$ be polynomial. 
	The category $\kl(\mathcal F_W(F + \Id))$ is enriched over \cpoj.
\end{proposition}
%\begin{proofsketch}
%	Follows by Lemma~\ref{lem:kl-fw-polynomial-continuous} and 
%	Proposition~\ref{prop:kl-fw-cpoj}.
%\end{proofsketch}
\begin{proof} %atend}
	By Lemma~\ref{lem:kl-fw-polynomial-continuous} and Theorem~\ref{thm:monadic-fg-zero-morph}, $\overline{F+\Id}$ is 
	locally continuous and part of a monad. % over $\kl(\mathcal F_W)$.
	We conclude by
	Proposition~\ref{prop:kl-fw-cpoj} and 
	Theorem~\ref{thm:t_kleisli_s_parameterised_saturation}.
\end{proof} %atend}

\looseness=-1
Consider the endofunctor $\mathcal F_W(A_\tau \times \Id)\cong
\mathcal{F}_W(A\times \Id +\Id)$ modelling Weighted LTSs with unobservable actions
\cite{gp:icalp2014,mp2013:weak-arxiv}.
The lifting of $A_\tau \times \Id$ is given, on
any $f \in \kl(\mathcal F_W)(X,Y)$, as:
\[
	(\overline{A_\tau \times \Id})f(a,x)(b,y) = \eta_{A_\tau}(a)(b)\cdot f(x)(y) 
\]
where $x \in X$, $y \in Y$, and $a,b \in A_\tau$.
By Theorem~\ref{thm:monadic-fg-zero-morph} the functor $\overline{A_\tau \times \Id}$ extends to a monad whose
unit $\theta$ is defined as $\theta_X(x) = (\mathsf{inr}\circ\eta_{(A_\tau \times X)})(\tau,x)$
and whose multiplication $\nu$ id defined, on each component $X \in\Set$, as
\begin{align*}
	\nu_{X}(a,\tau,x)(b,y) &=  \eta_{(A_\tau \times X)}{(a,x)}(b,y) &
	\nu_{X}(\tau,a,x)(b,y) &=  \eta_{(A_\tau \times X)}{(a,x)}(b,y) \\
	\nu_{X}(\tau,\tau,x)(b,y) &=  \eta_{(A_\tau \times X)}{(\tau,x)}(b,y) &
	\nu_{X}(a,b,x)(c,y) &= 0
\end{align*}
where $x,y \in X$ and $a,b,c \in A$.
Composition in $\kl(\overline{A_\tau \times \Id})$ follows 
from \eqref{diag:kl-s-composition}. Let $f\colon X \to Y$ and $g\colon Y \to Z$ be two maps in $\kl(\mathcal{F}_W(A_\tau\times \Id))$,
their composite $g\bullet f$ maps each $x \in X$ to
the weight function defined, on any $a \in A$ and $z \in Z$, as:
\begin{align*}
	(g\bullet f)(x)(\tau,z) &= \sum_{y\in Y}g(y)(\tau,z)\cdot f(x)(\tau,y)\\
	(g\bullet f)(x)(a,z) &= \sum_{y\in Y}g(y)(a,z)\cdot f(x)(\tau,y) + \sum_{y\in Y}g(y)(\tau,z)\cdot f(x)(a,y)\text{.}
\end{align*}

\begin{remark}
	The above expressions turn out to be those used in 
	\cite{gp:icalp2014,mp2013:weak-arxiv} to define and 
	compute weak bisimulations for weighted transition
	systems. 
\end{remark}

\paragraph{Weak behavioural equivalence}
Weak bisimulation for weighted transition systems was independently
studied in \cite{gp:icalp2014,mp2013:weak-arxiv}, covering Baier and
Hermann's weak bisimulation among others. Both works approach the
problem by means of recursive equations describing how unobservable
transitions are composed.  This yields a saturated system akin to the
linear equation systems in \cite{baier97:cav}. These equations depend
on the state space partition induced by the weak bisimulation relation
under definition.

In our settings, unobservables are hidden inside the arrow
composition in $\kl(\mathcal F_W(A_\tau \times \Id))$ and
then it is easy to see that \eqref{eq:PS} instantiates to the aforementioned 
recursive equations  from \cite{gp:icalp2014,mp2013:weak-arxiv}.
In fact, for any $\alpha\colon X \to \mathcal F_W(A_\tau \times X)$ 
and $f \colon X \to Y \in \cat{C}$ the equation 
$\beta = f^\sharp \vee \beta \circ \alpha$
expands into the equation system (with values in $W$, $x \in X$, $y \in Y$, and $a \in A$):
\begin{align*}
	\beta(x)(\tau,y) & = 
		f^\sharp(x)(\tau,y)
		\vee
		\sum_{z \in X} \alpha(x)(\tau,z)\cdot \beta(z)(\tau,y)\\[-3pt]
	\beta(x)(a,y) & = 
		f^\sharp(x)(a,y)
		\vee
		\sum_{z \in X} \alpha(x)(\tau,z)\cdot \beta(z)(a,y) + \sum_{z \in X} \alpha(x)(a,z)\cdot \beta(z)(\tau,y)
\end{align*}

Both works \cite{gp:icalp2014,mp2013:weak-arxiv} (with minor distinctions) define weak bisimulations 
as kernel pairs of morphisms being the least solution to the above 
equation(s) and, by Definition~\ref{def:weak-behavioural-equivalence} and
Theorem~\ref{thm:wbm-via-ps}, we have the following
correspondence:
\begin{proposition}
	\label{prop:adequacy-weighted-lts}
	For $\mathcal{F}_W(A_\tau \times \Id)$-coalgebras,
	weak behavioural equivalence corresponds to weak 
	weighted bisimulation \cite{gp:icalp2014}.
\end{proposition}
\begin{corollary}
	For $\mathcal{F}_{\mathbb N \cup \{\infty\}}(A_\tau \times \Id)$-coalgebras,
	weak behavioural equivalence coincides with 
	weak resource bisimulation \cite{ais:ipl2010}
	for systems weighted over the (complete)
	arithmetic semiring of	natural numbers
	$(\mathbb N \cup \{\infty\},+,\cdot)$.
\end{corollary}
\begin{proof}
	By Proposition~\ref{prop:adequacy-weighted-lts} and \cite{gp:icalp2014}.
\end{proof}

\subsection{Fully-probabilistic systems}\label{sec:ex-fully-probabilistic}
Fully-probabilistic systems are understood as $\mathcal{D}(A \times \Id)$-coalgebras (e.g.~see \cite{hasuo07:trace})
and are a special case of both weighted and Segala systems.

The Kleisli category for the (sub)distribution monad $\mathcal{D}_{\leq 1}$ 
is not \cpoj-enriched, as it lacks binary joins despite being
\cpo-enriched (see e.g.~\cite{hasuo07:trace}).
However, this monad can be embedded into monads
whose Kleisli categories are \cpoj-enriched, thus offering a setting where to apply our framework as we did for fully probabilistic systems. Examples of such monads
are $\CM $ and $\mathcal F_{[0,\infty]}$
which correspond to viewing fully-probabilistic systems
as (deterministic) Segala or weighted transition systems, respectively.  As a consequence, these two embeddings yield different equivalences; indeed,
it is well-known that Segala's weak (convex) bisimulation \cite{sl:njc95}
and Baier and Hermann's weak (probabilistic) bisimulation \cite{baier97:cav} do not coincide.
For instance, consider the fully-probabilistic system
\[
	\begin{tikzpicture}[auto]
		\node[circle,draw=black,inner sep=2pt] (x) at (0,0) {\(x\)};
		\node[circle,draw=black,inner sep=2pt] (y) at (3,0) {\(y\)};
		\node[circle,draw=black,inner sep=2pt] (z) at (1.5,0) {\(z\)};
		
		\draw[->] (x) to node {\(\tau,q\)} (z);
		\draw[->] (y) to node[swap] {\(\tau,p\)} (z);
		\draw[->,loop left] (x) to node {\(\tau,1-q\)} (x);
		\draw[->,loop right] (y) to node {\(\tau,1-p\)} (y);
	\end{tikzpicture}
\]
\looseness=-1
for some $0 \leq p < q \leq 1$: when seen as a (deterministic) Segala system, 
the greatest weak convex bisimulation is the identity relation, however
$x$ and $y$ are weakly bisimilar when compared using to Baier-Hermann's definition
-- since $x$ and $y$ can both reach the equivalence classes $\{z\}$ and 
$\{x,y\}$ both with probability $1$.

On the other hand, when fully-probabilistic systems are seen as transition systems weighted over $[0,\infty]$, 
every Baier-Hermann's weak probabilistic bisimulation is a weak 
weighted bisimulation and \emph{vice versa} \cite{gp:icalp2014,mp2013:weak-arxiv}.
\begin{proposition}
	\label{prop:adequacy-fully-probabilistic-lts}
	For $\mathcal{D}(A_\tau \times \Id)$-coalgebras,
	weak behavioural equivalence coincides with 
	Baier-Hermann's weak probabilistic bisimulation \cite{baier97:cav}.
\end{proposition}
\begin{proof}
	By Proposition~\ref{prop:adequacy-weighted-lts} 
	and	\cite{gp:icalp2014,mp2013:weak-arxiv}.
\end{proof}

\subsection{Nominal systems}
Presheaf categories are commonly used for modelling systems with
dynamically allocated resources, like names or memory regions \cite{ft:lics01}.  We will
show that \cpoj-enrichment (and hence saturation) extends from the 
Kleisli category of a monad over \Set to the Kleisli category of 
the monad lifted to the presheaf category of interest. As an example, we 
consider the late semantics of the $\pi$-calculus and precisely capture
its weak (late) bisimulation \cite{sangiorgi2003pi}.

\paragraph{Lifting monads to presheaves}
Let $\cat D$ be a small category and let $(T,\mu,\eta)$ be a monad on
\Set such that $\kl(T)$ is enriched over \cpoj.  We define the
extension of $T$ to the presheaf category $[\cat D,\Set]$ as the monad
given by $T^\cat DX \defeq T \circ X$, $\mu^\cat D_X \defeq \mu X$,
and $\eta^\cat D_X \defeq \eta X$.  Monad laws follow from
\Cat being a 2-category.

The category $\kl(T^\cat D)$ has presheaves on $\cat D$
as objects and natural transformations in $[\cat{D},\Set](-,T\circ-)$
as morphisms; hence the isomorphism holds:
\begin{equation*}
	\kl(T^\cat D) \cong [\cat D, \kl(T)].
\end{equation*}

\begin{proposition}
	If $\kl(T)$ is enriched over \cpoj so is $\kl(T^\cat D)$.
\end{proposition}
\begin{proof}
	Recall that $[\cat D,\kl(T)](X,Y) =[\cat{D},\Set](X,T\circ Y)$.
\end{proof}

Let $\mathbb I$ be a skeleton of the category of finite sets and
injective functions. The presheaf category $[\mathbb I, \Set]$
is the context of several works on coalgebraic semantics for
calculi with names; in fact they are strictly related to nominal sets
\cite{gmm:hosc05,pitts2013nominal}.
In particular, $\mathcal P^\mathbb I$ is precisely the component
expressing non-determinism in the behavioural functors
capturing the late and early semantics of the $\pi$-calculus
\cite{ft:lics01}.
Enrichment extends pointwise to $\kl(\mathcal P^\mathbb I)$
e.g.~$f \leq g \iff \forall n \in \mathbb I, f_n \leq g_n$.

Another example of systems with names is the Fusion
calculus where, differently from above, names can be unified
creating aliases. In \cite{miculan:mfps08} the second author presented a bialgebraic
account of Fusion using presheaves over a skeleton
of the category of finite sets and functions $[\mathbb
F,\Set]$.  Also in this case, the behavioural functor presents
a monadic component (besides observables) expressing non-deterministic
stateful computations.

\paragraph{The $\pi$-calculus}
Consider the $[\mathbb I,\Set]$-endofunctor 
\[
	B_\pi \defeq \mathcal{P}^\mathbb{I}(
		N \times N \times \Id + 
		N \times \Id^N + 
		N \times \delta +
		\Id)
\]
describing the late semantics for $\pi$-calculus \cite{ft:lics01};
here $\delta:[\mathbb I,\Set]\to [\mathbb I,\Set]$ is the \emph{dynamic allocation}
endofunctor, defined by $(\delta X)_n = X_{n+1}$; $N_n = n$,
and $(\mathcal P^{\mathbb I} X)_n = \mathcal P(X_n)$ 
is the extension of the powerset monad described above.
Note that exponentials $X^N$ in $[\mathbb I,\Set]$ are
not $N$-fold products. 
In fact, for any presheaf $X$ and stage $n \in \mathbb I$ we have:
\[
	B_\pi(X)_n = \mathcal P(n \times n \times X_n + 
	n \times (X_n)^n\times X_{n+1} + n \times X_{n+1} + X_n)
\]
describing the behaviour for processes with at most $n$ free names: the four
components of the coproduct describe output, input, bound input and
 $\tau$ transitions, respectively.

Let us define $F = F_o + F_i + F_b$, where $F_o=N \times N \times
\Id$, $F_i=N \times \Id^N$, $F_b=N \times \delta$;
then, $B_\pi = \mathcal P^{\mathbb I}(F+\Id)$.  The lifting
$\overline{F + \Id}$ to $\kl(\mathcal P^\mathbb I)$ is given,
for any $f \colon X \to Y \in \kl(\mathcal P^\mathbb I)$, by the cotuple:
\begin{align*}
	\overline{F_o}{f_n} &= \kstr{3} \circ (\eta_n \times \eta_n \times f_n)
	& 
        \overline{F_b}{f_n} & = n \times f_{n+1} \\		
	\overline{F_i}{f_n} &= 
		\kstr{(n+2)} \circ (\eta_n \times (f_n)^n \times f_{n+1})
	& \overline{\Id}{f_n} & = f_n
\end{align*}
where $\kstr{k}$ stands for a 
suitable $k$-fold strength of $\mathcal P$.
Note that all but $\overline{F_b}$ are canonical (at each stage) 
in the sense of \cite{hasuo07:trace}.

The functor $\overline{F+\Id}$ extends readily to a monad: its unit 
is $\mathsf{inr}\colon \Id \To \overline{F} + \Id$ 
and maps everything into the unobservable
component of the coproduct; its multiplication concatenates 
unobservable transitions discarding pairs of observables
by means of suitable zero morphisms as expected.
In particular, it is defined as:
\begin{align*}
	\nu_{X_n} (\tau,(a,b,x)) &= \{(a,b,x)\} &
	\nu_{X_n} (\tau,(a,\overline{x})) &= \{(a,\overline x)\} &
	\nu_{X_n} (\tau,(a,x)) &= \{(a,x)\}\\
	\nu_{X_n} (a,b,(\tau,x)) &= \{(a,b,x)\} &
	\nu_{X_n} (a,\overline{(\tau,x)}) &= \{(a,\overline x)\} &
	\nu_{X_n} (a,(\tau,x)) &= \{(a,x)\}\\
	\nu_{X_n} (\tau,(\tau,x)) & = \{(\tau,x)\} &
	\nu_{X_n} y & = \emptyset
\end{align*}
where $n \in \mathbb I$, $a,b \in N_n$ are names at $n$-th stage, $\overline{x},\overline{(\tau,x)}$ are $(n+1)$-tuples and $y$ covers the cases left out.

Composition in $\kl(\overline{F + \Id})$ follows directly
by \eqref{diag:kl-s-composition}. In particular, for any two compatible
morphisms $f \colon X \to Y$ and $g \colon Y \to Z$ their composite $(g \bullet f)$ 
maps each $x \in X_n$ to the set resulting from the union of the following sets:
\begin{align*}
	&\{(a,b,z) \mid  
		(a,b,y) \in f_n(x) \land (\tau,z) \in g_n(y) \text{ or }
		(\tau,y) \in f_n(x) \land (a,b,z) \in g_n(y)
	\}\\
	&\{(a,(z_1,\dots,z_{n+1})) \mid 
		(a,(y_1,\dots,y_{n+1})) \in f_n(x) {\land}
		(\tau,z_i) \in g_n(y_i) {\land} (\tau,z_{n+1}) \in g_{n+1}(y_{n+1})
	\}\\
	&\{(a,(z_1,\dots,z_{n+1})) \mid 
		(\tau,y) \in f_n(x) \land
		(a,(z_1,\dots,z_{n+1})) \in g_n(y)
	\}\\
	&\{(a,z) \mid 
		(a,y) \in f_n(x) \land
		(\tau,z) \in g_n(y) \text{ or }
		(\tau,y) \in f_n(x) \land
		(a,z) \in g_n(y)
	\}
	\\
	&\{(\tau,z) \mid 
		(\tau,y) \in f_n(x) \land
		(\tau,z) \in g_n(y)
	\}
\end{align*}
where $a,b$ are names at $n$-th stage. Using a more evocative
notation, composition is characterized by the following derivation rules:
\begin{gather*}
	\frac{
		x \xrightarrow{\tau}_{f_n} y \quad 
		y \xrightarrow{a!\langle b \rangle}_{g_n} z
	}{
		x \xrightarrow{a!\langle b \rangle}_{(g\bullet f)_n} z
	}
	\quad
	\frac{
		x \xrightarrow{a!\langle b \rangle}_{f_n} y \quad 
		y \xrightarrow{\tau}_{g_n} z
	}{
		x \xrightarrow{a!\langle b \rangle}_{(g\bullet f)_n} z
	}
	\quad
	\frac{
		x \xrightarrow{a!()}_{f_n} y \quad 
		y \xrightarrow{\tau}_{g_{n+1}} z
	}{
		x \xrightarrow{a!()}_{(g\bullet f)_n} z
	}
	\\
	\frac{
		x \xrightarrow{\tau}_{f_n} y \quad 
		y \xrightarrow{a!()}_{g_n} z
	}{
		x \xrightarrow{a!()}_{(g\bullet f)_n} z
	}
	\quad
	\frac{
		x \xrightarrow{\tau}_{f_n} y \quad 
		y \xrightarrow{a?()}_{g_n} z_1,\dots, z_n,z_{n+1}
	}{
		x \xrightarrow{a?()}_{(g\bullet f)_n} z_1,\dots, z_n,z_{n+1}
	}
	\\
	\frac{
		x \xrightarrow{a?()}_{f_n} y_1,\dots, y_n,y_{n+1} \quad 
		y_i \xrightarrow{\tau}_{g_n} z_i \quad
		y_{n+1} \xrightarrow{\tau}_{g_{n+1}} z_{n+1}
	}{
		x \xrightarrow{a?()}_{(g\bullet f)_n} z_1,\dots, z_n,z_{n+1}
	}
	\quad
	\frac{
		x \xrightarrow{\tau}_{f_n} y \quad 
		y \xrightarrow{\tau}_{g_n} z
	}{
		x \xrightarrow{\tau}_{(g\bullet f)_n} z
	}
\end{gather*}

\begin{proposition}
	\label{prop:pi-cpoj}
	The category $\kl(\mathcal P^\mathbb I(F +  \Id))$ is 
	\cpoj-enriched and satisfies left and right distributivity.
\end{proposition}
\begin{proof} %atend}
	By Theorem~\ref{thm:t_kleisli_s_parameterised_saturation} we only need
	to prove that $\overline{F + \Id}$ is locally continuous.
	The proof is quite straightforward since $F$ is polynomial except
	for $\Id^N$ and $\delta$. Recall that suprema of ascending $\omega$-chains in 
	$\kl(\mathcal P^\mathbb I)$ are given at each stage $n \in \mathbb I$
	as join in $\kl(\mathcal P)$. We conclude by noticing that at stage
	$n$ $\Id^N$ is polynomial:
	\[
	\textstyle
	\overline{\Id^N}\bigvee_{i < \omega} f_{i,n} 
	= \left(\bigvee_{i < \omega} f_{i,n}\right)^n \times \bigvee_{i < \omega} f_{i,n+1} 
	= \bigvee_{i < \omega}(f_{i,n})^n \times  f_{i,n+1}
	= \bigvee_{i < \omega} \overline{\Id^N}f_{i,n}
	\]
	for any and ascending $\omega$-chain $(f_i)_{i \in \mathbb N}$.
	By a similar argument, binary joins distribute over composition and $\overline{\delta}$ is locally continuous.
\end{proof} %atend}
The above construction can be easily adapted to many behaviours with
unobservables and resources modelled in some
presheaf category.

\paragraph{Weak (late) behavioural equivalence for the $\pi$-calculus}
As for LTSs and Segala systems, by Proposition~\ref{prop:pi-cpoj} 
we can apply directly
Theorem~\ref{theorem:ordered_saturation_category_LD_RD} and define
saturated systems by assigning each \emph{late transition system} 
to the least one closed under the following rules:
\begin{gather*}
	\frac{}{
		x \xRightarrow{\,\tau\ }_n x
	}
	\quad
	\frac{
		x \xRightarrow{\,\tau\ }_n x' \quad 
		x' \xrightarrow{\tau}_n y' \quad
		y' \xRightarrow{\,\tau\ }_n y
	}{
		x \xRightarrow{\,\tau\ }_n y
	}
	\\
	\frac{
		x \xRightarrow{\,\tau\ }_n x' \quad 
		x' \xrightarrow{a!\langle b \rangle}_n y' \quad
		y' \xRightarrow{\,\tau\ }_n y
	}{
		x \xRightarrow{\,a!\langle b \rangle\ }_n y
	}
	\quad
	\frac{
		x \xRightarrow{\tau\ }_n x' \quad 
		x' \xrightarrow{a!()}_n y' \quad
		y' \xRightarrow{\tau\ }_{n+1} y
	}{
		x \xRightarrow{a!()\ }_n y
	}
	\\
	\frac{
		x \xRightarrow{\,\tau\ }_n x' \quad 
		x' \xrightarrow{a?()}_n y'_1,\dots,y'_n,y'_{n+1} \quad
		y'_i \xRightarrow{\,\tau\ }_n y_i \quad
		y'_{n+1} \xRightarrow{\,\tau\ }_{n+1} y_{n+1}
	}{
		x \xRightarrow{\,a?()\ }_n y_1,\dots,y_n,y_{n+1}
	}
\end{gather*}
Here the single and double arrows denote the LTS for the given 
coalgebra and the saturated one, respectively, the stage $n \in \mathbb N$ is subscript,
$a,b \in n$, and every process is at stage $n$ except for $x_{n+1},y_{n+1}$ which are at stage $n+1$.

\begin{proposition}
	\label{prop:adequacy-pi-lts}
	For $B_\pi$-coalgebras,
	weak behavioural equivalence coincides with 
	weak (late) bisimulation \cite{sangiorgi2003pi}.
\end{proposition}
\begin{proof}
	The derivation rules describing saturation are exactly
	those presented in \cite{sangiorgi2003pi}. Then the correspondence
	follows by Theorems~\ref{theorem:ordered_saturation_category_LD_RD}, \ref{thm:wbm-via-ps}, and~\ref{thm:double-arrow}.
\end{proof}

\subsection{Topological Kripke frames}
In \cite{e:1974top-kripke} Esakia used the Vietoris topology to give a coalgebraic definition of topological Kripke frames, thus relating the 
Vietoris topology, modal logic and coalgebras; a connection that has been 
fruitfully investigated e.g.~in 
\cite{k:1981duality,r:1986vietoris,j:1985vietoris-loc,vv:2014vietoris-modal}.
Here we take into account unobservable moves.

\paragraph{Vietoris monad} 
Let $(X,\Sigma_X)$ be a compact Hausdorff space, and let $\mathcal{V}(X,\Sigma_X)$
the set of compact subsets of $X$.  The \emph{Vietoris topology}
$\Sigma_{\mathcal{V}(X,\Sigma_X)}$ on $\mathcal{V}(X,\Sigma_X)$ is described by the base
consisting of sets of the following form:
\[\textstyle
	\nabla\{U_1,\dots,U_n\} = 
	\left\{
		F \in \mathcal{V}(X,\Sigma_X)
	\,\middle|\, 
		F \subseteq \bigcup_{i =1}^n U_i
		\text{ and }
		F \cap U_i \neq \emptyset 
		\text{ for }
		1 \leq i \leq n
	\right\}
\]
where $n \in \mathbb N$ and each $U_i \in \Sigma_X$.
This extends to an endofunctor $\mathcal{V}$ over 
$\cat{KHaus}$, the category of compact Hausdorff spaces 
and continuous function, whose action takes forward images
$\mathcal{V} f(X') = f(X')$
for every continuous function $f \colon (X,\Sigma_X) \to (Y,\Sigma_Y)$
\cite{vv:2014vietoris-modal}.
This functor plays a central r\^ole in modal logics since Esakia's seminal
work on topological Kripke frames\footnote{%
	Topological Kripke frames are usually defined on Stone spaces.
} \cite{e:1974top-kripke}.
Moreover, it extends to a monad (mimicking the lifting of $\mathcal P$)
whose multiplication and unit are given, on each component $(X,\Sigma_X)$, as follows:
\[
	\mu_{(X,\Sigma_X)}(Y) = \bigcup Y
	\qquad
	\eta_{(X,\Sigma_X)}(x) = \{x\}
	\text{.}
\]
The structure of sets of compact subsets extends to
the Kleisli category for the Vietoris monad in the obvious way:
arrows are ordered and joined in pointwise manner. Kleisli
composition preserves this structure.
\begin{proposition}
	\label{prop:kl-vietoris-cpoj}
	The category $\kl(\mathcal{V})$ is enriched over \cpoj
	and satisfies left and right distributivity.
\end{proposition}
\begin{proof}
	Joins in $\mathcal{V}(X,\Sigma_X)$ (i.e., unions)
	are extended pointwise to $\kl(\mathcal{V})$.
\end{proof}

\paragraph{Vietoris transition systems}
The Vietoris monad on $\cat{KHaus}$ is commutative. Its
double strength is defined, on each component, as:
\[
	\dstr_{(X,\Sigma_X),(Y,\Sigma_Y)}(X',Y') = X' \times Y'\text{.}
\]
Coherence is proven as in the powerset case.
Therefore, polynomial functors lift canonically to $\kl(\mathcal{V})$.
\begin{lemma}
	\label{lem:kl-vietoris-polynomial-continuous}
	Canonical liftings of polynomial functors % to $\kl(\mathcal{V})$
	are locally continuous.\par
\end{lemma}
%\begin{proofsketch}
%	Follows by an argument similar to \cite[Lem.~2.6]{hasuo07:trace}.
%\end{proofsketch}
\begin{proof} %atend}
	The proof is carried out by structural recursion akin to
	Hasuo's proof of \cite[Lem.~2.6]{hasuo07:trace}. 
	Because the ordering
	is pointwise it suffices to show that $\mathsf{dstr}$ and 
	coprojections are continuous map between CPOs.

	Let $(X_i)_{i < \omega} \in \mathcal{V}(X,\Sigma_X)$ and 
	$(Y_i)_{i < \omega} \in  \mathcal{V}(Y,\Sigma_Y)$
	be two ascending $\omega$-chains. Continuity of
	$\dstr_{(X,\Sigma_X),(Y,\Sigma_Y)}(X',Y')$
	follows by \[\textstyle\bigvee_{i < \omega} X_i \times Y_i = 
	\left(\bigvee_{i < \omega} X_i\right)\times
	\left(\bigvee_{i < \omega} Y_i\right)\]
	and by definition of cartesian product of CPOs.
	Continuity of $T \mathsf{in}_i$ follows directly by 
	definition of coproducts in	$\cat{KHaus}$.	
\end{proof} %atend}
Continuity and existence of zero morphisms allow us to apply Theorem~\ref{thm:t_kleisli_s_parameterised_saturation} and
Theorem~\ref{thm:monadic-fg-zero-morph} to define a suitable monad
for functors like $\mathcal{V}(F + \Id)$.
\begin{proposition}
	Let $F\colon \cat{KHaus} \to \cat{KHaus}$ be a polynomial functor.
	The category $\kl(\mathcal{V}(F + \Id))$ is 
	enriched over \cpoj and satisfies \eqref{law:LD} and \eqref{law:RD}.
	\par
\end{proposition}
\begin{proof} %atend}
	Zero morphisms map everything to the bottom element of their co\-do\-main. 
	By Lemma~\ref{lem:kl-vietoris-polynomial-continuous} and Theorem~\ref{thm:monadic-fg-zero-morph}, $\overline{F+\Id}$ is locally continuous and part of a monad over $\kl(\mathcal{V})$.
	We conclude by
	Proposition~\ref{prop:kl-vietoris-cpoj} and 
	Theorem~\ref{thm:t_kleisli_s_parameterised_saturation}.
\end{proof} %atend}

In the following we instantiate the result on a CCS-like behaviour
where labels are drawn from a compact Hausdorff space $(A,\Sigma_A)$
plus a distinguished silent action $\tau$. 
Henceforth, let $F = (A,\Sigma_A) \times \Id$. 
The canonical lifting of $F + \Id$ on $\kl(\mathcal{V})$
is defined, on any continuous map $f \colon (X,\Sigma_X) \to (Y,\Sigma_Y)$, as:
\[
	\overline{F + \Id} = id_{(A,\Sigma_A)} \times f + f
	\text{.}
\]
By Theorem~\ref{thm:monadic-fg-zero-morph}, this functor carries a monad $(\overline{F + \Id},\nu,\theta)$ whose unit and multiplication are defined, 
on each compact Hausdorff space $(X,\Sigma_X)$, as:
\begin{align*}
	\nu_{(X,\Sigma_X)}(a,\tau,x) &=  \{(a,x)\} &
	\nu_{(X,\Sigma_X)}(\tau,a,x) &=  \{(a,x)\} &
	\theta_{(X,\Sigma_X)}(x) = \{(\tau,x)\}\\
	\nu_{(X,\Sigma_X)}(\tau,\tau,x) &=  \{(\tau,x)\} &
	\nu_{(X,\Sigma_X)}(a,b,x) &= \emptyset\text{.} &
\end{align*}
Composition in $\kl(\overline{F \times \Id})$ follows by
\eqref{eq:lts-kl-comp} and, for any two morphisms $f$ and $g$ with suitable
domain and codomain, the composite $g \bullet f$ is:
\[
	(g\bullet f)(x) = \{
		(a,z) 
		\mid
		(a,y) \in f(x) \land (\tau,z) \in g(y) \text{ or } 
		(\tau,y) \in f(x) \land (a,z) \in g(y) 
	\}
	\text{.}
\]

\paragraph{Weak behavioural equivalence}
Let $\alpha$ be a coalgebra for the functor
$\mathcal{V}((A,\Sigma_A)\times \Id + \Id)$
on $\cat{KHaus}$. The forgetful functor $U \colon \cat{KHaus} \to \Set$
extends to a forgetful\[
 U : \cat{KHaus}_{\mathcal{V}((A,\Sigma_A)\times \Id + \Id)}
 \to \Set_{\mathcal P(A\times \Id + \Id)}
 \text{.}
\]\par
\begin{proposition}
	The forgetful functor $U$ preserves weak bisimulation.
\end{proposition}
%\begin{proofsketch}
%	In both cases \eqref{law:LD} and \eqref{law:RD} hold. By
%	Theorems~\ref{theorem:ordered_saturation_category_LD_RD}, \ref{thm:wbm-via-ps}, and~\ref{thm:double-arrow}
%	$\alpha_i^{\ast} = i \circ \bigvee_{n < \omega} \alpha^n$.
%	Hence it suffices to show that $U$ is a \cpo-functor.
%\end{proofsketch}
\begin{proof} %atend}
	Both $\kl(\mathcal{V}((A,\Sigma_A)\times \Id + \Id))$
	and $\kl(\mathcal P(A\times \Id + \Id))$ satisfy \eqref{law:LD} and \eqref{law:RD}. By Theorems~\ref{theorem:ordered_saturation_category_LD_RD}, \ref{thm:wbm-via-ps}, and~\ref{thm:double-arrow}
	$\alpha_i^{\ast} = i\circ  \alpha^{\ast}$. Hence, it suffices to prove that 
	$U\alpha^{\ast} = (U\alpha)^{\ast}$ where $U$ is extended to 
	$\kl(\mathcal{V}((A,\Sigma_A)\times \Id + \Id))$ 
	in the obvious way.
	Recall that $\alpha^{\ast} = \bigvee_{n < \omega} \alpha^n$,
	the thesis follows by $U$ being continuous in the sense that,
	for any ascending $\omega$-chain $(f_i)_{i \in\mathbb N}$,
	$U\bigvee_{i<\omega} f_i = \bigvee_{i<\omega} Uf_i$.
	We conclude by noticing that $U$ forgets the topology of each space and
	that suprema of ascending $\omega$-chains in both cases are pointwise 
	suprema of ascending $\omega$-chains
	of compact subsets i.e.~countable unions.
\end{proof} %atend}
Reworded, every weak bisimulation for a Vietoris LTS is a weak bisimulation
for the non-deterministic LTS obtained forgetting its topological structure.

\subsection{Continuous-state stochastic systems}
Coalgebraic presentations of probabilistic and stochastic continuous-state systems received an increasing interest
e.g.~\cite{bm:2015stocsos,kerstan2013coalgebraic}. 
In this section we show how these types of systems fits into our framework.

\paragraph{Measures monad}
Let $(X,\Sigma_X)$ be a measurable space and let $\Delta (X,\Sigma_X)$ be
the set of all measures $\phi\colon \Sigma_X \to [0,\infty]$ on
$(X,\Sigma_X)$.  For each measurable set $M \in \Sigma_X$ there is a
canonical evaluation function $ev_M\colon \Delta(X,\Sigma_X) \to [0,\infty]$
such that $ev_M(\phi) = \phi(M)$; these evaluation maps allow us to
endow $\Delta (X,\Sigma_X)$ with the smallest $\sigma$-algebra rendering
each $ev_M$ measurable with respect to the Borel $\sigma$-algebra on
$[0,\infty]$ (i.e.~the initial $\sigma$-algebra w.r.t.~$\{ev_M \mid M
\in \Sigma_X\}$).  This definition extends to an endofunctor $\Delta$
over $\cat{Meas}$, the category of measurable spaces and measurable
functions, acting on any $(X,\Sigma_X)$ and $f \colon (X,\Sigma_X) \to
(Y,\Sigma_Y)$ as:
\[
	\Delta(X,\Sigma_X) \defeq
		(\Delta(X,\Sigma_X),\Sigma_{\Delta(X,\Sigma_X)})
	\qquad
	\Delta f (\phi) \defeq \phi \circ f^{-1}
\]

\begin{lemma}
	\label{lem:delta-monad}
	The functor $\Delta \colon \cat{Meas} \to \cat{Meas}$ extends to a monad $(\Delta,\mu,\eta)$ 
	whose multiplication and unit are given, on each 
	component $(X,\Sigma_X)$, as follows:
	\begin{align*}
		\mu_{(X,\Sigma_X)}(\phi)(M) 
		&= \int_{\Delta(X,\Sigma_X)}\psi(M)\phi(\diff \psi)
		=
		\int_{\Delta(X,\Sigma_X)}ev_M\diff \phi
		\\
		\eta_{(X,\Sigma_X)}(x)(M) &= \delta_x(M) = \chi_M(x)
	\end{align*}
	where $\chi_M \colon X \to \{0,1\}$ is the indicator function on $M \in \Sigma_X$.
\end{lemma}
\begin{proofatend}
	Let $(X,\Sigma_X)$ be a measurable space.
	A function from $\Sigma_X$ to $[0,\infty]$ is called 
	simple whenever it is a finite linear
	combination of indicator functions of $\Sigma_X$-measurable sets.
	For any measurable function $f : \Sigma_X \to [0,\infty]$
	there is a monotonic increasing sequence of non-negative simple 
	functions $(f_n)_{n < \omega}$ such that 
	$f = \lim_{n \to \infty} f_n$. In particular, such a sequence
	can be obtained by defining each $f_n$ as 
	$\sum_{i=1}^{n\cdot2^n} \frac{i}{2^n}\chi_{N_{n,i}}$
	where $N_{n,i}$ is the $\Sigma_X$-measurable
	$\left\{x \in X \,\middle|\, f(x) \in \left[\frac{i}{2^n},\frac{i+1}{2^n}\right)\right\}$
	for $i < n\cdot2^n$ and
	$\left\{x \in X \,\middle|\, n \leq f(x)\right\}$
	for $i = n\cdot2^n$.

	For any $(X,\Sigma_X)$, $\phi \in \Delta(X,\Sigma_X)$ 
	and	$M \in \Sigma_X$ we have that:
	\begin{align*}
		&(\mu_{(X,\Sigma_X)}\circ\eta_{\Delta(X,\Sigma_X)})(\phi)(M)
		=
%		\int_{\Delta(X,\Sigma_X)}\psi(M)\eta_{\Delta(X,\Sigma_X)}(\phi)(\diff \psi)
%		=
		\int_{\Delta(X,\Sigma_X)}ev_M\diff \eta_{\Delta(X,\Sigma_X)}(\phi) 
		\\=\;&
		\sup\left\{
				\int_{\Delta(X,\Sigma_X)}\! f\diff \eta_{\Delta(X,\Sigma_X)}(\phi)
			\,\middle|\!
				\begin{array}{l}
					0\leq f\leq ev_M,\\
					 f \text{ is simple}
				\end{array}\!\!\!
			\right\}
		=
		\lim_{n \to \infty}\left\{
			\sum_{i=1}^{n2^n} \frac{i}{2^n}\chi_{N_{n,i}}(\phi)
		\right\} 
%		\\=\;&
		=
		\phi(M)
	\end{align*}
	where each $N_{n,i}$ is the $i$-th measurable set of the $n$-th simple
	function defining $ev_M$ as above. Likewise:
	\begin{align*}
		&(\mu_{(X,\Sigma_X)}\circ\Delta\eta_{(X,\Sigma_X)})(\phi)(M)
		= \int_{\Delta(X,\Sigma_X)}ev_M\diff \Delta\eta_{(X,\Sigma_X)}(\phi)
		\\=\;&
		\sup\left\{
				\int_{\Delta(X,\Sigma_X)}\! f\diff \eta_{\Delta(X,\Sigma_X)}(\phi)
			\,\middle|\!
				\begin{array}{l}
					0\leq f\leq ev_M,\\
					 f \text{ is simple}
				\end{array}\!\!\!
			\right\}
		=
		\lim_{n \to \infty}\left\{
				\sum_{i=1}^{n2^n} \frac{i}{2^n}\phi\circ \eta_{(X,\Sigma_X)}^{-1}(N_{n,i})
			\right\}
		\\=\;& \phi(M)
	\end{align*} 
	This completes the proof that $\eta$ satisfies the unit laws of the monad.	
	
	For any measurable function $f \colon X \to [0,\infty]$
	assume that:
	\begin{equation}
		\label{eq:doberkat-lemma-3.2.2}
		\int_{(X,\Sigma_X)} f \diff  \mu_{(X,\Sigma_X)}(\phi) = 
		\int_{\Delta(X,\Sigma_X)}\left(\int_{(X,\Sigma_X)} 
		f \diff  \rho\right) \phi (\diff  \rho)
		\text{.}
	\end{equation}
	
	For any measurable space $(X,\Sigma_X)$, measure $\phi \in \Delta(X,\Sigma_X)$ 
	and	measurable set $M \in \Sigma_X$, the multiplication law for the monad
	is:
	\begin{align*}
		&(\mu_{(X,\Sigma_X)}\circ\mu_{\Delta(X,\Sigma_X)})(\phi)(M)
		= \int_{\Delta(X,\Sigma_X)}ev_M\diff \mu_{\Delta(X,\Sigma_X)}(\phi) 
		\\\eqlabel[\text{i}]{eq:delta-mult-1}\;& \int_{\Delta^2(X,\Sigma_X)}\left(\int_{\Delta(X,\Sigma_X)} 
				ev_M \diff  \rho\right) \phi (\diff  \rho)
		= \int_{\Delta(X,\Sigma_X)}ev_M\circ \mu_{(X,\Sigma_X)}\diff \phi
		\\\eqlabel[\text{ii}]{eq:delta-mult-2}\;& 
			\int_{\Delta^2(X,\Sigma_X)}ev_M\diff \phi\circ \mu_{(X,\Sigma_X)}^{-1}
		= (\mu_{(X,\Sigma_X)}\circ\Delta\mu_{(X,\Sigma_X)})(\phi)(M)
	\end{align*}
	where \eqref{eq:delta-mult-1} follows from \eqref{eq:doberkat-lemma-3.2.2}
	and \eqref{eq:delta-mult-2} from the \emph{change of variables theorem}\footnote{See e.g.~T.~Tao, \emph{An Introduction to Measure Theory} (2011)}.
			
	Finally, we prove \eqref{eq:doberkat-lemma-3.2.2}.
	The argument follows Doberkat's proof of the same equality
	in the case of bounded measurable functions and the Giry monad
	(cf.~\cite[Lem~3.2.2]{doberkat:09stoc-logic}). If $f$ is the
	indicator function $\chi_M \colon X \to [0,\infty]$ (for $M \in \Sigma_X$) 
	then \eqref{eq:doberkat-lemma-3.2.2} is precisely $\mu_{(X,\Sigma_X)}$.
	By linearity of integration of simple functions \eqref{eq:doberkat-lemma-3.2.2}
	holds for any non-negative simple function.	Since every non-negative 
	measurable function is the limit of some monotone increasing sequence 
	of non-negative simple functions, \eqref{eq:doberkat-lemma-3.2.2} holds,
	by \emph{monotone convergence theorem}, on any non-negative measurable 
	function completing the proof.
\end{proofatend}
Roughly speaking, $\Delta$ can be though as the 
``measurable equivalent'' of $\mathcal F_{[0,\infty]}$
and, just as $\mathcal F_{[0,\infty]}$ generalises the 
probability distribution monad $\mathcal D$, $\Delta$  
generalises the probability measure monad $\Delta_{=1}$
(a.k.a.~\emph{Giry monad} \cite{prakash:02bisim,prakashbook09,doberkat:09stoc-logic}).

Each measurable space $\Delta(X,\Sigma_X)$ presents a partial
order induced by the pointwise extension of the natural order on $[0,\infty]$.
For any two measures $\phi$ and $\psi$ on $(X,\Sigma_X)$ we have:
$
	\phi \leq \psi \defiff \forall M \in \Sigma_X . \phi(M) \leq \psi(M)
	\text{.}
$
The join of any $\phi,\psi \in \Delta(X,\Sigma_X)$ 
is given, on each measurable set $M\in \Sigma_X$, by:
	\[
		\phi \vee \psi(M)\defeq\sup\{ \phi(M_1) + \psi(M_2) \mid
		M_1\in \Sigma_X, M_1 \subseteq M, 
		M_2 = M \setminus M_1 \}.
	\]
Any increasing $\omega$-chain $(\phi_i)_{i<\omega}$ in
$\Delta(X,\Sigma_X)$ has a supremum and it is given,
on each $M \in \Sigma_X$, as:
$
	{\textstyle\bigvee_{i < \omega}\phi_i(M)} \defeq 
		\lim_{i \to \infty} \phi_i(M) =
		\sup_{i} \{\phi_i(M)\}
	\text{.}
$
Each $\Delta(X,\Sigma_X)$ is an $\omega$-Cpo with binary joins.
This structure extends pointwise to the hom-objects of $\kl(\Delta)$
and is preserved by composition in $\kl(\Delta)$.
\begin{proposition}
	\label{prop:kl-delta-cpoj}
	The category $\kl(\Delta)$ is \cpoj-enriched.
	\par
\end{proposition}
\begin{proofatend}
	The proof is organized as follows:
	we show that 
	\begin{enumerate*}[label=\em(\alph*)]
		\item
			each $\Delta(X,\Sigma_X)$ has binary joins and 
			is an $\omega$-Cpo;
		\item
			each hom-set of $\kl(\Delta)$ has binary joins and 
			is an $\omega$-Cpo;
		\item
			composition is continuous in both components.
	\end{enumerate*}
	Therefore $\kl(\Delta)$ is \cpo-enriched, has binary joins.

	\paragraph{\emph{(a)}} The join of any $\phi,\psi \in \Delta(X,\Sigma_X)$ 
	is given, on $M\in \Sigma_X$, by:
	\[
		\phi \vee \psi(M)\defeq\sup\left\{ \phi(M_1) + \psi(M_2) \mid
		M_1\in \Sigma_X, M_1 \subseteq M, 
		M_2 = M \setminus M_1 \right\}\text{.}
	\]
	In order to verify this claim first we need to show that $\phi\vee \psi$ is a measure. Indeed, $\phi\vee \psi(\emptyset) = 0$. Consider a countable disjoint family of measurable sets $M_i\in \Sigma_X$ and let $M=\bigcup_i M_i$. Assume $\phi\vee \psi(M) < \infty$ and $\phi\vee \psi(M_i) < \infty$. In this case, for any $\varepsilon>0$ there are measurable, disjoint sets $N_1,N_2$ s.t.~$M=N_1\cup N_2$ and
	\begin{align*}
		 &\phi\vee \psi (M) - \varepsilon < 
		 \phi(N_1)+\psi(N_2) = 
		 \phi\left(\bigcup_{i=1}^{\infty} M_i \cap N_1\right) + \psi\left(\bigcup_{i=1}^{\infty} M_i \cap N_2\right) 
		 \\=\;&
		 \sum_{i=1}^{\infty} \phi(M_i\cap N_1) +\psi(M_i\cap N_2) \leq 
		 \sum_{i=1}^{\infty} \phi\vee\psi(M_i)\text{.}
	\end{align*}
	The above is true for any $\varepsilon > 0$ and hence $\phi\vee \psi (\bigcup_i M_i) \leq \sum_i \phi\vee\psi(M_i)$. To see the reverse inequality is true consider arbitrary $\varepsilon$ and note that for any $M_i$ there are measurable, disjoint sets $N_{1,i}, N_{2,i}$ such that $M_i = N_{1,i} \cup N_{2,i}$ we have:
	\[
	\phi\vee \psi(M_i) - \frac{\varepsilon}{2^i} \leq \phi(N_{1,i}) + \psi(N_{2,i})\text{.}
	\]
	Hence, 
	\[
		\sum_{i=1}^{\infty} \phi \vee \psi(M_i) - \varepsilon \leq \sum_{i=1}^{\infty} \phi(N_{1,i}) + \psi(N_{2,i})
		=
		\phi\left(\bigcup_{i=1}^{\infty} N_{1,i}\right)+\psi\left(\bigcup_{i=1}^{\infty} N_{2,i}\right) 
		\leq \phi\vee \psi (M)\text{.}
	\]
	Hence, $\sum_{i=1}^{\infty} \phi \vee \psi(M_i) \leq  \phi\vee \psi (M)$. Let $\phi\vee \psi(M) = \infty$; then for any $r>0$ there are disjoint and measurable sets $N_1,N_2$ such that $M=N_1\cup N_2$ and 
	\[
	r\leq \phi(N_1) + \psi(N_2) = \sum_{i=1}^{\infty} \phi(M_i\cap N_1) +\psi(M_i\cap N_2)\leq \sum_{i=1}^{\infty}\phi\vee\psi(M_i)\text{.}
	\]
	Hence,  $\sum_{i=1}^{\infty} \phi\vee\psi(M_i)= \infty =  \phi\vee \psi (M)$. Similarly we verify the case for $ \phi\vee\psi(M_i) = \infty$ for some $i$. This proves $\sigma$-additivity of $\psi\vee \phi$. In order to complete the proof we will show that $\psi\vee \phi$ is indeed the supremum of $\psi$ and $\phi$. Consider any measure $\rho:\Sigma\to [0,\infty]$ such that $\phi,\psi\leq \rho$. For any measurable $M$ and any measurable and disjoint $N_1,N_2$ such that $M=N_1\cup N_2$ we have:
	\[
		\rho(M) = \rho(N_1\cup N_2) = \rho(N_1)+\rho(N_2)\geq \phi(N_1)+\psi(N_2)\text{.}
	\]
	Hence, $\rho(M) \geq \phi \vee \psi(M)$. Therefore, the hom-sets of $\kl(\Delta)$ admit arbitrary binary joins, which follows by the fact that the order on these hom-sets is the pointwise extension of the order on measures.
	
	Any increasing $\omega$-chain $(\phi_i)_{i\in \mathbb N}$ in
	$\Delta(X,\Sigma_X)$ has a supremum and it is given,
	on each measurable set $M \in \Sigma_X$, as:
	\[
		\bigvee_{i < \omega}\phi_i(M) \defeq 
			\lim_{i \to \infty} \phi_i(M) =
			\sup_{i < \omega} \{\phi_i(M)\}
		\text{.}
	\]
	We need to verify that $\bigvee_{i<\omega} \phi_i:\Sigma_X\to [0,\infty]$ is indeed a measure. We have $\bigvee_{i<\omega} \phi_i(\emptyset)  = 0$ and for a countable family of measurable sets $M_n\in \Sigma_X$ the following is true:
	\begin{align*}
		&\bigvee_{i<\omega} \phi_i\left(\bigcup_{n<\omega} M_n\right) = 
		\sup_{i<\omega} \left\{ \phi_i\left(\bigcup_{n<\omega} M_n\right) \right\} 
		= 
		\sup_{i<\omega} \left\{ \sum_{n<\omega} \phi_i(M_n) \right\} 
		\\=\;&
		\sup_{i<\omega} \sup_{n<\omega} \left\{ \sum_{k=1}^n \phi_i(M_k) \right\} =
		\sup_{n<\omega} \sup_{i<\omega} \left\{ \sum_{k=1}^n \phi_i(M_k) \right\} = 
		\sup_{n<\omega} \left\{\sum_{k=1}^n \sup_i \phi_i(M_k) \right\}
		\\=\;&
		\sup_{n<\omega} \left\{\sum_{k=1}^n \bigvee_{i<\omega} \phi_i(M_k)\right\} 
		=
		\sum_{n<\omega}  \bigvee_{i<\omega} \phi_i(M_k)\text{.} 
	\end{align*}
	This proves that $\bigvee_{i<\omega} \phi_i$ is a measure on $(X,\Sigma_X)$.
	
	\paragraph{\emph{(b)}} Binary joins of measures are extended to 
	hom-objects of $\kl(\Delta)$ in a pointwise manner:
	$(f \vee g)(x) \defeq f(x) \vee g(x)$
	for any $f$ and $g$ in $\kl(\Delta)((X,\Sigma_X),(Y,\Sigma_Y))$
	and $x \in X$. In order to verify the claim we have to show that 
	the function $f \vee g$ is measurable. Since $\Sigma_{\Delta(Y,\Sigma_Y)}$
	is the initial $\sigma$-algebra w.r.t.~$\{ev_M \mid M \in \Sigma_X\}$
	it suffices to show that $ev_M \circ (f \vee g)$ is measurable for 
	any $M \in \Sigma_X$. 
	The join $(f\vee g)$ is pointwise and defined as the 
	supremum of measurable functions:
	\[
	(f\vee g)(x)(M) =\sup\left\{
			(ev_{M_1}\circ f)(x) + (ev_{M_2}\circ g)(x)
		\,\middle|\,
		\begin{array}{l}
			M_1 \in \Sigma_X,
			M_1 \subseteq M,\\
			M_2 = M \setminus M_1
		\end{array}
		\right\}
	\]
	Hence, the claim follows by showing that the above supremum
	can always be reformulated into one over a countable set.
	This set is formed by a suitable family of measurable functions
	${ev_{M^i_1}\circ f + ev_{M^i_2}\circ g}$ parameterised by
	a countable sequence of suitable pairs $(M^i_1,M^i_2)$.
	In fact, the supremum is taken over a subset of
	$[0,\infty]$ and hence can always be defined as the limit of
	some non-decreasing $\omega$-chain of its elements.
		
	Suprema for ascending $\omega$-chains of measurable functions
	are defined by pointwise extension of suprema for ascending
	$\omega$-chains of measures i.e.~for any ascending $\omega$-chain
	$(f_i)_{i \in \mathbb N}$ in $\kl(\Delta)((X,\Sigma_X),(Y,\Sigma_Y))$
	and $x \in X$:
	\[\textstyle
		\big(\bigvee_{i < \omega} f_i\big) (x) \defeq 
		\bigvee_{i < \omega} f_i(x) = \sup\{f_i(x)\}
		\text{.}
	\]
	Measurability follows by $\Sigma_{\Delta(X,\Sigma_X)}$ being
	the initial $\sigma$-algebra with respect to $\{ev_M \mid M \in \Sigma_X\}$
	and by suprema being given pointwise:
	\[
		\big({\textstyle\bigvee_{i < \omega}} f_i\big) (x)(M) = 
		\sup\{f_i(x)(M)\} = \lim_{i\to \infty} f_i(x)(M)
		\text{.}
	\]
		
	\paragraph{\emph{(c)}} Finally, we show that composition is continuous in both components i.e.~su\-pre\-ma are preserved.
	Let $(f_i)_{i \in \mathbb N}$ be an ascending $\omega$-chain in $\kl(\Delta)$
	as above and let $h \colon (W,\Sigma_W) \to (X,\Sigma_X)$ and 
	$k \colon (Y,\Sigma_Y) \to (Z,\Sigma_Z)$ be arrows in $\kl(\Delta)$.
	Composition is left continuous:
	${\textstyle\bigvee} f_i \bullet h = ({\textstyle\bigvee} f_i) \bullet h$.
	In fact, the following holds:
	\begin{align*}
		&(\mu_{(Z,\Sigma_Z)} \circ \Delta ({\textstyle\bigvee_{i<\omega}} f_i) \circ h)(x)(M)
		= \int_{\Delta(Z,\Sigma_Z)} ev_M
			\diff  h(x)\circ ({\textstyle\bigvee_{i<\omega}} f_i)^{-1}
		\\=\;&\int_{\Delta(Y,\Sigma_Y)} ev_M \circ {\textstyle\bigvee_{i<\omega}} f_i
			\diff  h(x)
		\eqlabel[\text{i}]{eq:kl-delta-left-continuous-2}
			\int_{\Delta(Y,\Sigma_Y)} {\textstyle\bigvee_{i<\omega}} ev_M \circ f_i
			\diff  h(x)
		\\\eqlabel[\text{ii}]{eq:kl-delta-left-continuous-3}\;& {\bigvee}\int_{\Delta(Y,\Sigma_Y)} ev_M \circ f_i
			\diff  h(x)
		= ({\textstyle\bigvee_{i<\omega}} \mu_{(Z,\Sigma_Z)} \circ \Delta f_i \circ h)(x)(M)
	\end{align*}
	where \eqref{eq:kl-delta-left-continuous-3} follows from the monotonic convergence theorem and \eqref{eq:kl-delta-left-continuous-2} from
	\[
		(ev_M\circ {\textstyle\bigvee} f_i)(x) = 
		({\textstyle\bigvee} f_i)(x)(M) = 
		\sup\{f_i(x)(M)\} = ({\textstyle\bigvee} ev_M\circ f_i)(x)\text{.}
	\]
	The proof that 
	${\textstyle\bigvee} g \bullet f_i = g \bullet {\textstyle\bigvee} f_i$
	is completely analogous, and this completes the proof.
\end{proofatend}

\paragraph{Measurable transition systems}
As for the Giry monad it generalise, the measure monad $\Delta$ is equipped with a 
strength, a costrength, and a double strength. In particular, the latter is given
on each component as:
\[
	\dstr_{(X,\Sigma_X),(Y,\Sigma_Y)}(\phi,\psi)(M \times N) =
	(\phi \otimes \psi)(M \times N) = \phi(M)\cdot\psi(N)
\text{.}\]\par
\begin{lemma}
	The measure monad $(\Delta,\mu,\eta)$ is strong and commutative.
\end{lemma}
\begin{proofatend}
	Let $\mathcal{X} = (X,\Sigma_X)$ and $\mathcal{Y} = (Y,\Sigma_Y)$
	be two measurable spaces. The proof proceeds as follows:
	\begin{enumerate*}[label=\em(\alph{*})]
		\item
	we show that 
	$\str_{\mathcal{X},\mathcal{Y}}(x,\psi) = \delta_x \otimes \psi$
	defines a strength for $\Delta$;
		\item
	we define a costrength out of $\str$ and the cartesian structure of \cat{Meas};
		\item
	we show $\str$ and $\cstr$ render $\Delta$ commutative and
	define
	$\dstr_{\mathcal{X},\mathcal{Y}}(\phi,\psi) = \phi \otimes \psi$.
	\end{enumerate*}
	
	The natural transformation $\str$ is coherent with respect to $\mu$:
	\begin{align*}
		&\left(\mu_{\mathcal{X}\times\mathcal{Y}} %\circ
			\Delta\left(\str_{\mathcal{X},\mathcal{Y}}\right) %\circ
			\str_{\mathcal{X},\Delta\mathcal{Y}}\right)(x,\psi)(M\times N) 
		= 
%		\int_{\Delta((X,\Sigma_X)\times(Y,\Sigma_Y))} 
%			\rho(M\times N)\cdot(\Delta(\str_{(X,\Sigma_X),(Y,\Sigma_Y)})\circ
%			\str_{(X,\Sigma_X),\Delta(Y,\Sigma_Y)})(x,\psi)(\diff \rho) 
%		=
%		\int_{\Delta((X,\Sigma_X)\times(Y,\Sigma_Y))} 
%			ev_{M\times N} \diff\Delta(\str_{(X,\Sigma_X),(Y,\Sigma_Y)})\circ
%			\str_{(X,\Sigma_X),\Delta(Y,\Sigma_Y)})(x,\psi)	
%		\\=
		\int_{\Delta(\mathcal{X}\times\mathcal{Y})} 
			ev_{M\times N} \diff(\delta_x \otimes \psi)\circ
			\str_{\mathcal{X},\mathcal{Y}}^{-1}
		\\=\;&
		\lim_{n \to \infty}\left\{
			\sum_{i=0}^{n2^n} \frac{i}{2^n}
			\cdot
			(\delta_x \otimes \psi)
			\left(\left\{(x',\phi) 
			\;\middle|\; 
			x' \in M \land i \leq 2^n \cdot \psi(N) < i+1
%			\in \left[\frac{i}{2^n},\frac{i+1}{2^n}\right)
			\right\}\right)
		\right\}
		\\=\;&
		\lim_{n \to \infty}\left\{
			\sum_{i=0}^{n2^n} \frac{i}{2^n}
			\cdot
			\delta_x(M)\cdot \psi(\{\phi \mid i \leq 2^n \cdot \phi(N) < i+1
			\})
		\right\}
		\\=\;&
		\delta_x(M)\cdot \lim_{n \to \infty}\left\{
			\sum_{i=0}^{n2^n} \frac{i}{2^n}
			\cdot
			\psi(\{\phi \mid i \leq 2^n \cdot \phi(N) < i+1
			\})
		\right\}
		\\=\;&
		\delta_x(M)\cdot \int_{\Delta\mathcal{Y}} ev_{N} \diff \psi
		= 
		\delta_x(M)\cdot \mu_{\mathcal{Y}}(\psi)(N)
		=
		\left(\str_{\mathcal{X},\mathcal{Y}} \circ \left(id_{\mathcal{X}} \times \mu_{\mathcal{Y}}\right)\right)(M\times N)
	\end{align*}
	and with respect to $\eta$:
	\[
		\eta_{\mathcal{X},\mathcal{Y}}(x,y)
		= \delta_{(x,y)}
		= \delta_x \otimes \delta_y
		= 
		\str_{\mathcal{X},\mathcal{Y}}(x,\eta_{\mathcal{Y}}(y))
		\text{.}
	\]
	The remaining coherence conditions readily follow by similar arguments and
	routine expansion of the definitions. 
	
	Cartesian products define a symmetric monoidal structure on 
	$\cat{Meas}$; henceforth let $\gamma$ denote the corresponding braiding
	natural isomorphism. As for any strong monad on a symmetric monoidal category,
	$\gamma$ and $\str$ induce the costrngth $\cstr$ given, on each component,
	as follows:
	\[
		\cstr_{\mathcal{X},\mathcal{Y}}(\phi,y) \defeq 
		\left(\Delta\gamma_{\mathcal{Y},\mathcal{X}} \circ 
		\str_{\mathcal{Y},\mathcal{X}} \circ 
		\gamma_{\Delta\mathcal{X},\mathcal{Y}}\right)(\phi,y) = 
		\phi \otimes \delta_y
		\text{.}
	\]
	Finally, $\str$ and $\cstr$ define the double strength above:
	\begin{align*}
		& \left(\mu_{\mathcal{X},\mathcal{Y}}\circ
			\Delta\left(\str_{\mathcal{X},\mathcal{Y}}\right)\circ
			\cstr_{\mathcal{X},\Delta\mathcal{Y}}\right)
			(\phi,\psi)(M\times N)
		\\=\;&
		\int_{\Delta(\mathcal{X}\times\mathcal{Y})} 
			ev_{M\times N} \diff\left(\cstr_{\mathcal{X},\Delta\mathcal{Y}}(\phi,\psi)\right)\circ
			\str_{\mathcal{X},\mathcal{Y}}^{-1}
		\\=\;&
		\lim_{n \to \infty}\left\{
			\sum_{i=0}^{n2^n} \frac{i}{2^n}
			\cdot
			(\phi \otimes \delta_\psi)
			\left(\left\{(x',\psi') 
			\;\middle|\; 
			x' \in M \land i \leq 2^n \cdot \psi'(N) < i+1
			\right\}\right)
		\right\}
		\\=\;&
		\lim_{n \to \infty}\left\{
			\sum_{i=0}^{n2^n} \frac{i}{2^n}
			\cdot
			\phi(M) \cdot \delta_\psi
			\left(\left\{\psi' 
			\;\middle|\; 
			i \leq 2^n \cdot \psi'(N) < i+1
			\right\}\right)
		\right\}	
		\\=\;&
		\phi(M) \cdot \psi(N) = (\phi \otimes \psi)(M \times N)\text{.}
	\end{align*}
	By a symmetric argument, we have that:
	\[
		\left(\mu_{\mathcal{X},\mathcal{Y}}\circ
				\Delta\left(\str_{\mathcal{X},\mathcal{Y}}\right)\circ
				\cstr_{\mathcal{X},\Delta\mathcal{Y}}\right)
				(\phi,\psi) = \phi \otimes \psi
	\]
	completing the proof.
\end{proofatend}
Therefore, polynomial functors have canonical liftings to $\kl(\Delta)$.
\begin{lemma}
	\label{lem:kl-delta-polynomial-continuous}
	Canonical liftings of polynomial functors
	are locally continuous.\par
\end{lemma}
%\begin{proofsketch}
%	Follows by an argument similar to \cite[Lem.~2.6]{hasuo07:trace}.
%\end{proofsketch}
\begin{proof} %atend}
	The proof is carried out by structural recursion akin to
	Hasuo's proof of \cite[Lem.~2.6]{hasuo07:trace}. 
	Because the ordering
	is pointwise it suffices to show that $\mathsf{dstr}$ and coprojections
	are continuous map between CPOs.

	Let $(\phi_i)_{i <\omega}$ and $(\psi_i)_{i < \omega}$
	be two ascending $\omega$-chains in the spaces $\Delta(X,\Sigma_X)$ and $\Delta(Y,\Sigma_Y)$,
	respectively. Continuity of $\dstr_{(X,\Sigma_X),(Y,\Sigma_Y)}(\phi,\psi)$
	follows by 
	\[\textstyle
	\bigvee_{i < \omega} \phi_i(M)\cdot\psi_i(N) = \left(
	\bigvee_{i < \omega} \phi_i(M)\right)\cdot\left(\bigvee_{i < \omega}\psi_i(N)\right)
	\]
	and by definition of cartesian product of CPOs.
	Continuity of $\Delta \mathsf{in}_i$ follows directly by 
	definition of coproducts in	$\cat{Meas}$.	
\end{proof} %atend}
Continuity and existence of zero morphisms allow us define monads for 
functors like $\Delta(F + \Id)$ and s.t.~their Kleisli
categories are enriched over \cpoj.

\begin{proposition}
	Let $F$ be a polynomial endofunctor on $\cat{Meas}$.
	The category $\kl(\Delta(F + \Id))$ is \cpoj-enriched.
\end{proposition}
\begin{proof} %atend}
	Zero morphisms map everything to the 
	zero measure.
	By Lemma~\ref{lem:kl-delta-polynomial-continuous} and Theorem~\ref{thm:monadic-fg-zero-morph} $\overline{F+\Id}$ is 
	locally continuous and part of a monad over $\kl(\Delta)$. 
	We conclude by Proposition~\ref{prop:kl-delta-cpoj} and 
	Theorem~\ref{thm:t_kleisli_s_parameterised_saturation}.
\end{proof} %atend}

In the following we instantiate the result on ``measurable LTS''
like those described e.g.~by FlatCCS \cite{bm:2015stocsos}: a calculus
for CCS-like processes living in the Euclidean plane. To this end,
let $(A,\Sigma_A)$ be the chosen measurable space of labels and
let $F = (A,\Sigma_A) \times \Id$. 
The canonical lifting of $F + \Id$ is defined, on every 
$f \colon (X,\Sigma_X) \to (Y,\Sigma_Y) \in \kl(\Delta)$, as:
\[
	((\overline{F + \Id})f)(a,x)(M \times N) = \chi_{M}(a)\cdot f(x)(N)\text{.}
\]
Unit and multiplication of the canonical monad $(\overline{F + \Id},\nu,\theta)$
are defined, on each measurable space $(X,\Sigma_X)$ as
$\theta_{(X,\Sigma_X)}(x) = \mathsf{inr}\circ\delta_{(\tau,x)}$
and as
\begin{align*}
	\nu_{(X,\Sigma_X)}(a,\tau,x)(M) &=  \chi_M(a,x) &
	\nu_{(X,\Sigma_X)}(\tau,a,x)(M) &=  \chi_M(a,x) \\
	\nu_{(X,\Sigma_X)}(\tau,\tau,x)(M) &=  \chi_M(\tau,x) &
	\nu_{(X,\Sigma_X)}(a,b,x)(M) &= 0
\end{align*}
respectively.
Composition in $\kl(\overline{F + \Id})$ follows directly
by \eqref{diag:kl-s-composition}. In particular, for any two 
morphisms $f \colon (X,\Sigma_X) \to (Y,\Sigma_Y)$ and $g \colon (Y,\Sigma_Y) \to (Z,\Sigma_Z)$
the composite $(g \bullet f)$ is defined as follows, where $S = F + \Id$:
\[
	(g \bullet f)(x)(M) = 
	\int_{\Delta S(Z,\Sigma_Z)} \psi(M) f(x)\circ 
	(\mu_{S(Z,\Sigma_Z)}\circ \Delta\nu_{(Z,\Sigma_Z)}\circ \overline Sg)^{-1}
	(\diff \psi)
	\text{.}
\]

\paragraph{Weak behavioural equivalence}
A \emph{measurable transition systems} with labels in $(A,\Sigma_A)$
and unobservable moves is an endomorphism in the Kleisli category
for the monad $(\overline{F+\Id},\nu,\theta)$ defined above.
$\Delta((A,\Sigma_A)\times\Id + \Id)$. 

By instantiating Definition~\ref{def:weak-behavioural-morphism-ps} in this setting,
it is clear that a measure function $i\colon (X,\Sigma_X) \to (Y,\Sigma_Y)$
is a weak behavioural morphism for a measurable LTS $\alpha$
whenever there exist a measurable LTS $\beta$ such that
$\alpha^\ast_{i^\sharp}  = \beta\bullet i^\sharp = \beta \circ i$.
This allows us to instantiate Definition~\ref{def:weak-behavioural-equivalence} obtaining the following new behavioural equivalence for this kind of systems.
\begin{definition}
	Let $\alpha$ be a measurable LTS carried by $(X,\Sigma_X)$.
	A relation $(R,\Sigma_R) \rightrightarrows (X,\Sigma_X)$ 
	is a \emph{weak (measurable) behavioural equivalence} for $\alpha$ whenever $p$,
	the canonical projection to the quotient space, 
	is a weak behavioural morphisms
	for $\alpha$ i.e.:
	\[\alpha^\ast_{p^\sharp}(x) = \alpha^\ast_{p^\sharp}(x')\]
	for every $x,x' \in X$ such that $p(x) = p(x')$.
\end{definition}
Although the refining system $\beta$ is not mentioned explicitly in the above
definition, it is uniquely induced by $\alpha^\ast_{p^\sharp}$ and $p$ as:
\[
	\beta(p(x))(M) \defeq \alpha^\ast_{p^\sharp}(x)(M)
	\text{.}
\]
The endomorphism $\beta$ is well-defined since $p$ is the canonical projection
of the weak behavioural equivalence $(R,\Sigma_R) \rightrightarrows (X,\Sigma_X)$.
Moreover, since \cat{Meas} has epimorphism-embedding factorisation system
whose epis are also regular \cite{bm:2015stocsos} we can safely 
consider $J$ as the category of all regular epis in \cat{Meas}.

Since binary joins are given pointwise, we conclude by Lemma~\ref{lem:pointwise-jrd} and Theorem~\ref{thm:weak-is-complete-wrt-strong}
that every strong behavioural morphism is also a weak one.

\section{Conclusions and future work}
\label{sec:concl}

In this paper we have introduced a general framework for the definition
of behavioural equivalences for coalgebras with unobservable moves.
This framework is based on categories with a suitable order
enrichment; this structure allows us to saturate systems with respect to their
unobservable moves so that we cannot distinguish them on the base of
how many unobservable moves are needed to mimic an opponent move.
Remarkably, Kleisli categories of many monads used in concurrency
theory to define behavioural functors fit this setting: as we have
showed, these include the powerset monad, the convex combinations
monad, the generalised multiset (i.e., weighting) monad, the Vietoris
monad, etc. Using this framework, we have provided a general abstract
definition of weak behavioural equivalence; we have showed that this
notion covers several ``weak bisimilarities'' defined in literature
(for weighted LTS, Segala systems, calculi with names, etc.), just by
choosing the corresponding behavioural functor.  Notably, these results
apply also to categories different from \Set; e.g., we have also
considered presheaves, compact Hasdorff spaces and measurable spaces.

The benefits of this framework are manifold.  First, it
provides a general, uniform way for defining behavioural equivalences
in presence of unobservable moves; this allows us to readily obtain
new (weak) behavioural equivalences for new kinds of systems, even in categories like
$[\mathbb{I},\Set]$, $\cat{KHaus}$ or $\cat{Meas}$. 

Secondly, it has a ``normative'' value, in the sense that we can
characterise behavioural equivalences by looking at how
their saturation is defined. For instance, we can tell when a 
a variant of  Milner's double arrow construction or a variant of Baier-Hermann's saturation is used.

Furthermore, from a more fundamental point of view our work points out
the r\^ole of left and right distributivity in many properties of weak
behavioural equivalences, such as correctness and equivalence
with respect to strong behavioural equivalences over ``double-arrow'' saturated
systems.

Finally, our framework singles out the computational aspect of
``unobservability'' from the definition of system behaviours, thus
allowing for a modular development of the theory.  In fact, although
in this paper we have focused on weak behavioural equivalences, the
same approach can be adapted to other behavioural equivalences
just by changing the monadic structure
of (un)observable actions or the form of saturation.

\paragraph{Related work}
Several authors have previously investigated weak bisimulations under
the lens of coalgebraic theory; we mention
\cite{rothe2002towards,sokolova09:sacs}, the (independent) works from
the authors of this paper \cite{brengos2013:corr,mp2013:weak-arxiv},
and more recently \cite{gp:icalp2014}.

Likely, the closest to ours is \cite{gp:icalp2014}, where Goncharov
and Pattinson presented an account of weak bisimulation for
$T(A_\tau \times \Id)$-coalgebras for $T$ being a \Set-monad whose
Kleisli category is enriched over $\cat{dCPO}_\perp$, the category of
\emph{directed CPOs with bottom}.  The major difference between this
work and \cite{gp:icalp2014} is the approach to labels and internal
moves.  In fact, \cite{gp:icalp2014} builds on \emph{observational
  patterns} (i.e.~subsets of $A^{\ast}$ closed under Brzozowski
derivatives) as the driver of the recursive equations defining
saturation; on the other hand, in this work we treat internal moves as
a computational effect added to the ``behavioural stack'' to express
unobservability.  In \textit{loc.~cit.}, Goncharov and Pattinson
investigate the use of other binary operations on $T$ besides
pointwise joins, indicating \emph{algebraicity}
\footnote{A natural transformation $\rho \colon T^n \To T$ is \emph{algebraic} iff for any $f \colon X \to TY$ we have $\rho_Y\circ(f^\ddagger)^n = f^\ddagger \circ \rho_X$ where  $f^\ddagger = \mu_Y \circ Tf$.} 
as a condition for capturing weak
bisimulations as strong ones over saturated systems (i.e.~Milner's
double arrow construction).  Algebraicity of $\vee$ is sufficient but
not necessary to satisfy \eqref{law:LD} and, moreover, it implicitly
assumes $\vee$ to be pointwise which, again, is sufficient but not
necessary to satisfy \eqref{law:JRD}.  Interestingly,
\cite{gp:icalp2014} introduces a way to ``complete'' the continuous
function space monad $(-,D)_c$ (\emph{mutatis mutandis}, $\mathcal{F}_D$
where $D$ is a directed CPO with bottom) rendering $\vee$ algebraic; this
construction yields the ``continuous continuation monad''
$((-,D),D)_c$ and allows them to recover Baier-Hermann's weak
bisimulation as a strong bisimulation over a saturated system
(generalising \cite{sokolova09:sacs} where fully-probabilistic systems
are turned into suitable $\mathcal{DP}(A_\tau\times\Id)$-coalgebras,
for $\mathcal{DP} = ((-,2),[0,1])_c$).  In
Section~\ref{sec:2_cat_perspective} we have given an abstract
construction for achieving $\eqref{law:LD}$.
The generality of this construction is guaranteed by
Theorem~\ref{thm:saturation-independent}: any alternative
``\ref{law:LD}-completion'' is equivalent when it comes to
saturation, hence it yields the same weak behavioural
equivalence. Still, it would be of interest to study which completions
render $\vee$ algebraic.

Fixed point constructions in order enriched categories are the abstract
setting of many works investigating systems with ``structures on
labels'' from the coalgebraic perspective.  We mention the influential
account of trace semantics given in \cite{hasuo07:trace} and the
recent and neat understanding of silent $\varepsilon$-moves developed
in \cite{bonchi2015killing}.  The latter work can be seen as the
``dual'' approach of the present paper: while $\tau$-moves have to be
added, $\varepsilon$-moves have to be removed.
So, silent transitions are born from a star and killed with a dagger.

\paragraph{Final remarks and open questions}

The exact connection between \eqref{law:LD} and saturation is
still unclear.  Although we know (by Theorem
\ref{theorem:ordered_saturation_category_LD_RD}) that the latter
follows from former, we do not know whether there are some
reasonable and natural examples of categories which admit saturation
without left distributivity.  It is unclear what the weak
behavioural equivalence would be in this setting.

\paragraph{Acknowledgements}We thank the organizers Marco Bernardo, Daniel Gebler and Michele Loreti and all participants of the \emph{Open Problems in Concurrency
  Theory} seminar for useful discussions about a preliminary
version of these results, and the anonymous reviewers for their valuable suggestions and remarks on the journal version of this work.

\ifappendix
%	\clearpage
	\appendix
	\omittedproofs[\section{Longer proofs}\label{proofatend:proofs}]
\fi
\end{document}